\documentclass[11pt,a4paper,english]{article}
\usepackage[T1]{fontenc}
\usepackage[latin9]{inputenc}
\usepackage{color}
\usepackage{babel}
\usepackage{textcomp}
\usepackage{amsthm}
\usepackage{amsmath}
\usepackage{amssymb}
\usepackage{graphicx}
\usepackage{esint}
\usepackage[unicode=true,pdfusetitle,
 bookmarks=true,bookmarksnumbered=false,bookmarksopen=false,
 breaklinks=false,pdfborder={0 0 1},backref=section,colorlinks=false]
 {hyperref}
\usepackage{breakurl}

\makeatletter

\numberwithin{equation}{section}
\numberwithin{figure}{section}
\numberwithin{table}{section}
\theoremstyle{plain}
\newtheorem{thm}{\protect\theoremname}
  \theoremstyle{definition}
  \newtheorem{defn}[thm]{\protect\definitionname}
  \theoremstyle{plain}
  \newtheorem{conjecture}[thm]{\protect\conjecturename}

\usepackage{a4wide}
\usepackage{psfrag}
\usepackage{graphicx}
\usepackage{xcolor}
\usepackage{import}

\@ifundefined{showcaptionsetup}{}{%
 \PassOptionsToPackage{caption=false}{subfig}}
\usepackage{subfig}
\makeatother

  \providecommand{\conjecturename}{Conjecture}
  \providecommand{\definitionname}{Definition}
\providecommand{\theoremname}{Theorem}

\begin{document}

\title{Exact finite volume expectation values of local operators in excited
states}

\author{B. Pozsgay$^{1}$, I.M. Szécsényi$^{2,3}$ and G. Takács$^{1,4}$\\
~\\
$^{1}$MTA-BME \textquotedbl{}Momentum\textquotedbl{} Statistical
Field Theory Research Group\\
1111 Budapest, Budafoki út 8, Hungary\\
~\\
 $^{2}$Department of Mathematical Sciences, \\
Durham University \\
 South Road, Durham, DH1 3LE, United Kingdom\\
~\\
$^{3}$Institute of Theoretical Physics,\\
Eötvös Loránd University\\
1117 Budapest, Pázmány Péter sétány 1/A, Hungary\\
~\\
 $^{4}$Department of Theoretical Physics, \\
 Budapest University of Technology and Economics\\
1111 Budapest, Budafoki út 8, Hungary}

\date{29th December 2014}
\maketitle
\begin{abstract}
We present a conjecture for the exact expression of finite volume
expectation values in excited states in integrable quantum field theories,
which is an extension of an earlier conjecture to the case of general
diagonal factorized scattering with bound states and a nontrivial
bootstrap structure. The conjectured expression is a spectral expansion
which uses the exact form factors and the excited state thermodynamic
Bethe Ansatz as building blocks. The conjecture is proven for the
case of the trace of the energy-moment tensor. Concerning its validity
for more general operators, we provide numerical evidence using the
truncated conformal space approach. It is found that the expansion
fails to be well-defined for small values of the volume in cases when
the singularity structure of the TBA equations undergoes a non-trivial
rearrangement under some critical value of the volume. Despite these
shortcomings, the conjectured expression is expected to be valid for
all volumes for most of the excited states, and as an expansion above
the critical volume for the rest. 
\end{abstract}

\section{Introduction}

Finite temperature expectation values play an important role in various
applications of quantum field theory, and have been intensively studied
in recent years in the context of $1+1$ dimensional integrable quantum
field theories. For the one-point function a series expansion was
conjectured by Leclair and Mussardo in \cite{Leclair:1999ys}, based
on the thermodynamic Bethe Ansatz \cite{Yang:1968rm} as applied to
integrable quantum field theories \cite{Zamolodchikov:1989cf}, and
the exact form factors from the bootstrap program \cite{Karowski:1978vz,Kirillov:1987jp,Smirnov:1992vz}.
Recently, the Leclair-Mussardo series found applications in quantum
quenches \cite{2010NJPh...12e5015F,Pozsgay2011} and the investigations
of one-dimensional quantum gases \cite{Kormos:2009eqa,Kormos:2010rg,Pozsgay:2011ec}. 

Another method to obtain finite temperature correlators is based on
finite temperature form factors \cite{Doyon:2005jf,2014JSMTE..09..021C};
so far, however, this approach seems limited to free theories such
as the Ising model. Other approaches that can be used to construct
finite temperature expectation values in integrable models use separation
of variables \cite{Lukyanov:2000jp,2012JSMTE..10..006G}, or exploit
the hidden Grassmannian/fermionic structure of the XXZ spin chain
\cite{2011LMaPh..96..325J,2013NuPhB.875..166N,2014IJMPA..2950111N}. 

In \cite{Pozsgay2008a,Pozsgay2008b} a description of form factors
in finite volume was introduced, which was subsequently used to prove
the Leclair-Mussardo series \cite{Pozsgay2011}. This formalism was
also used to compute the finite-temperature two-point function \cite{Essler:2007jp,Essler:2009zz},
for which a systematic expansion was developed in \cite{Pozsgay2010a,Szecsenyi2012}.
Besides further applications to two-point functions in condensed matter
systems \cite{2012PhRvB..85a4402T,2014arXiv1403.7222W}, the finite
volume form factor formalism have found numerous other applications,
in computing one-point functions in the presence of boundaries \cite{Kormos:2010ae},
in the study of quantum quenches in field theories \cite{2013PhRvL.111j0401M,2014PhLB..734...52S,2014JSMTE..10..035B},
and in the context of holographic duality \cite{2014JPhA...47e5401K,2014JHEP...09..050B}.
They also provide a useful tool for testing exact form factor solutions
obtained from the form factor bootstrap, recently in the boundary
\cite{2009NuPhB.807..625B} and defect \cite{2014NuPhB.882..501B}
settings. 

The finite volume form factor formalism introduced in \cite{Pozsgay2008a,Pozsgay2008b}
only included the corrections that decay with a power of the volume;
exponential corrections were neglected. However, soon after \cite{Pozsgay2008a,Pozsgay2008b}
a method was proposed in \cite{Pozsgay2008} to construct certain
exponential corrections, the so-called $\mu$-terms which are related
to the bootstrap fusion between the particles. It was found that these
can be very important in determining matrix elements and resonance
parameters \cite{Takacs:2011nb}. 

However, for integrable quantum field theories one expects that an
exact determination of finite volume form factors is also possible.
For the Ising model on a lattice this was known before \cite{Bugrij:2000is,Bugrij:2001nf};
however, until recently there have been no such results for generic
models.

In \cite{Pozsgay2013} an extension of the Leclair-Mussardo series
was conjectured to describe exact excited state expectation values
(a.k.a. diagonal form factors) in finite volume. The methods of \cite{Pozsgay2013}
apply to theories like the sinh-Gordon model which have no bound states
in their bootstrap, and the thermodynamic Bethe Ansatz equations describing
excited states in finite volume have a particularly simple structure
\cite{2008NuPhB.799..403T}. 

In the present work we extend this conjecture to theories with a diagonal
factorized scattering that have a nontrivial bootstrap structure.
In such models, the excited state levels in finite volumes are described
by excited TBA systems of the type introduced in \cite{Bazhanov:1996aq,Dorey:1996re,Dorey:1997rb}.
This conjecture is verified in two ways. First, we show for the trace
of stress-energy tensor the conjectured series is equivalent to the
result obtained directly from the thermodynamic Bethe Ansatz. Second,
we make use of the truncated conformal space approach (TCSA) \cite{Yurov:1989yu}
to get a nontrivial further check of the series. For the latter, we
use the methods developed in \cite{Szecsenyi:2013gna}, to which the
interested reader is referred to for details. 

The outline of the paper is as follows. In Section \ref{sec:conjecture}
we introduce our notations and state the conjecture. In section \ref{sec:Theta-TBA-LMEx}
we present the proof that the conjectured series gives the same result
as the thermodynamic Bethe Ansatz when evaluated for the trace of
the stress-energy tensor. In Section \ref{sec:Finite-volume-expectation-in-T2}
we turn to the so-called $T_{2}$ model used as testing ground, and
specify the expansion for the case of the excited state thermodynamic
Bethe Ansatz of this particular field theory. The resulting expectation
value are then compared to numerical results from the TCSA in Section
\ref{sec:Numerical-results}, while Section \ref{sec:Conclusions-and-outlook}
contains our conclusions and outlook. As the method used for numerically
evaluating the connected diagonal form factors contains some non-trivial
tricks, and could be useful for other applications, it is presented
in Appendix \ref{sec:form_factors}.

\section{Finite volume expectation values in excited states: the conjecture
\label{sec:conjecture}}

The Leclair-Mussardo series for the finite volume vacuum expectation
value of a local operator in an integrable model with diagonal scattering
and $k$ species of massive particles takes the following form \cite{Leclair:1999ys}:
\begin{eqnarray}
\left\langle \mathcal{O}\right\rangle _{L} & = & \sum_{n_{1},\dots,n_{k}=0}^{\infty}\left(\prod_{i=1}^{k}n_{i}!\right)^{-1}\int_{-\infty}^{\infty}\prod_{j=1}^{\tilde{N}}\frac{\mathrm{d}\theta_{j}}{2\pi\left[1+e^{\varepsilon_{\beta_{j}}\left(\theta_{j}\right)}\right]}F_{2n_{1},\dots,2n_{k},c}^{\mathcal{O}}\left(\theta_{1},\dots,\theta_{\tilde{N}}\right)\label{eq:vacuum_LM}
\end{eqnarray}
where $F_{2n_{1},\dots,2n_{k},c}^{\mathcal{O}}$ are the connected
diagonal form factors of the operator \emph{$\mathcal{O}$}, $n_{i}$
the number of particles of species $i$, $\tilde{N}=\sum n_{i}$ the
total number of particles, and the $j$th particle has rapidity $\theta_{j}$
and species $\beta_{j}$. The $\varepsilon_{\alpha}(\theta)$ are
the pseudo-energy functions satisfying the TBA integral equation \cite{Zamolodchikov:1989cf}:

\begin{eqnarray}
\varepsilon_{\alpha}\left(\theta\right) & = & m_{\alpha}L\cosh\left(\theta\right)-\sum_{\beta}\int\frac{\mathrm{d}\theta'}{2\pi}\varphi_{\alpha\beta}\left(\theta-\theta'\right)\log\left(1+e^{-\varepsilon_{\beta}\left(\theta'\right)}\right)\label{eq:vacuum_TBA}
\end{eqnarray}
where the kernels are given by the logarithmic derivatives of the
two-particle scattering phases 
\begin{equation}
\varphi_{\alpha\beta}\left(\theta\right)=-i\frac{\partial}{\partial\theta}\log S_{\alpha\beta}\left(\theta\right)\label{eq:phidef}
\end{equation}
The finite volume ground state energy is given by
\begin{equation}
E_{\mathrm{TBA}}(L)=E\left(L\right)-\mathcal{B}L=-\sum_{\beta}\int\frac{\mathrm{d}\theta'}{2\pi}m_{\beta}\cosh\left(\theta\right)\log\left(1+e^{-\varepsilon_{\beta}\left(\theta\right)}\right)\label{eq:vacuum_TBA_energy}
\end{equation}
where $\mathcal{B}$ is the bulk energy density. The connected diagonal
form factors are defined by regularizing the diagonal matrix element
\begin{equation}
F^{\mathcal{O}}(\theta_{n}+i\pi+\epsilon_{n},\dots,\theta_{1}+i\pi+\epsilon_{1},\theta_{1},\dots,\theta_{n})\label{eq:regdiagff}
\end{equation}
and retaining the terms which are independent of the ratios $\epsilon_{i}/\epsilon_{j}$.
We remark that the form factor has a finite, but direction dependent
limit when all the $\epsilon_{i}$ are taken to zero simultaneously,
so the regularized matrix elements can only depend on their ratios.

From the work by Dorey and Tateo \cite{Dorey:1996re,Dorey:1997rb}
it is known that starting from the TBA equation of the ground state
one can reach the Riemann surface of excited states by analytic continuation
in the volume parameter; the same equations were also obtained in
\cite{Bazhanov:1996aq} using a different approach. When performing
the analytic continuation singularities of the $\log\left(1+e^{-\varepsilon_{\beta}}\right)$
terms, corresponding to locations where $Y_{\beta}=e^{\varepsilon_{\beta}}=-1$,
cross the integration contour modifying the TBA equations as

\begin{eqnarray}
\varepsilon_{\alpha}\left(\theta\right) & = & m_{\alpha}L\cosh\left(\theta\right)-\sum_{i=1}^{N}\eta_{i}\log S_{\alpha\alpha_{i}}\left(\theta-\bar{\theta}_{i}\right)\nonumber \\
 &  & -\sum_{\beta}\int\frac{\mathrm{d}\theta'}{2\pi}\varphi_{\alpha\beta}\left(\theta-\theta'\right)\log\left(1+e^{-\varepsilon_{\beta}\left(\theta'\right)}\right)\nonumber \\
E_{\mathrm{TBA}}\left(L\right) & = & \sum_{i=1}^{N}im_{\alpha_{i}}\eta_{i}\sinh\left(\bar{\theta}_{i}\right)-\sum_{\beta}\int\frac{\mathrm{d}\theta}{2\pi}m_{\beta}\cosh\left(\theta\right)\log\left(1+e^{-\varepsilon_{\beta}\left(\theta\right)}\right)\label{eq:excited_state_TBA_and_energy}
\end{eqnarray}
where the $\bar{\theta}_{i}$ are the positions of the singularities
of the pseudo-energy of species $\alpha_{i}$, which satisfy the quantization
conditions

\begin{eqnarray}
\varepsilon_{\alpha_{i}}\left(\bar{\theta}_{i}\right) & = & i\pi\left(2n_{i}+1\right)\qquad n_{i}\in\mathbb{Z}\label{eq:excited_TBA_quantization_conditions}
\end{eqnarray}
and the $n_{i}$ can be viewed as quantum numbers specifying the excited
state. Such singularities are called \emph{active}; their contribution
further depends on the orientation $\eta_{i}$ of the integration
contour around the singularity, which take the values $\pm1$ if the
singularity crossed the real axis from above/below, respectively.
The number and type of the active singularities depends on the excited
state and the position of these singularities at a fixed volume $L$
fully specifies the excited state. Therefore the corresponding finite
volume state can be denoted as
\begin{equation}
\left|\bar{\theta}_{i},\dots,\bar{\theta}_{N}\right\rangle _{L}
\end{equation}
The term $\log\left(1+e^{-\varepsilon_{\beta}}\right)$ also has singularities
where the $Y_{\beta}=e^{\varepsilon_{\beta}}=0$. In \cite{Dorey:1996re,Dorey:1997rb}
it was shown that whenever a singularity that corresponds to a zero
of a $Y$ function crosses the integration contour, it does not generate
new source terms to the TBA equations, but only rearranges the active
singularities already present. The only exception is when such a singularity
pinches the integration contour; for more details about dealing with
this situation see Subsection \ref{sub:desingularisation}. 

The TBA system (\ref{eq:excited_state_TBA_and_energy}) can be recast
in a universal functional form called the $Y$-system \cite{Zamolodchikov:1991et,Ravanini:1992fi}
\begin{eqnarray}
Y_{\alpha}\left(\theta-\frac{i\pi}{h}\right)Y_{\alpha}\left(\theta+\frac{i\pi}{h}\right) & = & \prod_{\beta=1}^{k}\left(1+Y_{\beta}(\theta)\right)^{I_{\alpha\beta}}\label{eq:Ysystem}\\
Y_{a}(\theta) & = & e^{\epsilon_{a}(\theta)}\nonumber 
\end{eqnarray}
where $h$ is the Coxeter number and $I_{\alpha\beta}$ is the incidence
matrix of some diagram. In Subsection \ref{sub:desingularisation}
we shall use the fact that $Y$-system relates the positions of the
two types of logarithmic singularities of the TBA equations. 

The analytical continuation is expected to connect not only the energy,
but also other quantities such as e.g. expectation values corresponding
to the different finite volume levels. It was shown in \cite{Pozsgay2013}
how to perform the residue integrals over the modified contours and
re-sum the terms into a compact form for an analytically continued
Leclair-Mussardo conjecture. This calculation was carried out for
the sinh-Gordon theory, where the excited TBA system is still a conjecture
\cite{2008NuPhB.799..403T}, but the result passes several consistency
checks. Namely, the first $e^{-mR}$ corrections in the infrared limit
agree with theoretical expectations and the result also agrees with
the TBA results for the trace of the stress-energy tensor. 

Let us now state the conjecture for the general form of the finite
volume expectation values in excited states. It contains two kind
of quantities, the ``dressed version'' of the diagonal form factors
and the densities of the active singularities. 
\begin{defn}
\label{def:dressed-diagonal-ff} The \emph{dressed diagonal form factors
of the local operator} $\mathcal{O}$ are 
\end{defn}
\begin{eqnarray}
\mathcal{D_{\varepsilon}^{O}}\left(\bar{\theta}_{1},\dots,\bar{\theta}_{l}\right) & := & \sum_{n_{1},\dots,n_{k}=0}^{\infty}\frac{1}{\prod_{i}n_{i}!}\int_{-\infty}^{\infty}\prod_{j=1}^{\tilde{N}}\frac{\mathrm{d}\theta_{j}}{2\pi\left[1+e^{\varepsilon_{\beta_{j}}\left(\theta_{j}\right)}\right]}\nonumber \\
 &  & \times F_{2l,2n_{1},\dots,2n_{k},c}^{\mathcal{O}}\left(\bar{\theta}_{1},\dots,\bar{\theta}_{l},\theta_{1},\dots,\theta_{\tilde{N}}\right)\label{eq:dressed_FF}
\end{eqnarray}
where $\bar{\theta}_{i}$ are a subset of the active singularities,
with the $i$th one corresponding to species $\alpha_{i}$. 

To obtain the densities of the active singularities, consider the
derivative matrix with respect to the singularity positions
\begin{eqnarray}
\mathcal{K}_{jk} & = & \frac{\partial Q_{j}}{\partial\bar{\theta}_{k}}\label{eq:TBA_Jacobian}
\end{eqnarray}
of the quantization conditions (\ref{eq:excited_TBA_quantization_conditions}) 

\begin{eqnarray}
Q_{j} & = & i\eta_{j}\varepsilon_{\alpha_{j}}\left(\bar{\theta}_{j}|\bar{\theta}_{1},\dots,\bar{\theta}_{N}\right)\label{eq:TBA_Qfunction}
\end{eqnarray}
satisfied by the position of the active singularities. 
\begin{defn}
\label{def:density-of-active-sing} The \emph{density of active singularities}
(in rapidity space) is the determinant of the derivative matrix
\begin{eqnarray}
\rho\left(\bar{\theta}_{1},\dots,\bar{\theta}_{N}\right) & = & \det\mathcal{K}_{ij}\label{eq:TBA_density}
\end{eqnarray}

\end{defn}
\ 
\begin{defn}
\label{def:restricted-density-of-active-sing} For any bipartite partition
$\left\{ \bar{\theta}_{1},\dots,\bar{\theta}_{N}\right\} =\left\{ \bar{\theta}_{+}\right\} \cup\left\{ \bar{\theta}_{-}\right\} $
of the active singularities, the \emph{restricted density of active
singularities in the subset} $\left\{ \bar{\theta}_{+}\right\} $
\emph{relative to $\left\{ \bar{\theta}_{-}\right\} $ }is defined
by 
\begin{equation}
\rho\left(\left\{ \bar{\theta}_{+}\right\} |\left\{ \bar{\theta}_{-}\right\} \right)=\det\mathcal{K}_{+}\label{eq:TBA_restricted_density}
\end{equation}
where $\mathcal{K}_{+}$ is the submatrix corresponding to the subset
of active singularities $\left\{ \bar{\theta}_{+}\right\} $. 
\end{defn}
Using the above definitions, the main result can be stated as follows:
\begin{conjecture}
The exact finite volume expectation values of an operator $\mathcal{O}$
in any finite volume state can be written as
\begin{eqnarray}
_{L}\left\langle \bar{\theta}_{i},\dots,\bar{\theta}_{N}\right|\mathcal{O}\left|\bar{\theta}_{i},\dots,\bar{\theta}_{N}\right\rangle _{L} & = & \frac{1}{\rho\left(\bar{\theta}_{1},\dots,\bar{\theta}_{N}\right)}\nonumber \\
 &  & \times\sum_{\left\{ \bar{\theta}_{+}\right\} \cup\left\{ \bar{\theta}_{-}\right\} }\mathcal{D_{\varepsilon}^{O}}\left(\left\{ \bar{\theta}_{+}\right\} \right)\rho\left(\left\{ \bar{\theta}_{-}\right\} |\left\{ \bar{\theta}_{+}\right\} \right)\label{eq:LMExcited_general}
\end{eqnarray}

\end{conjecture}
This conjecture was verified by explicit calculation for the case
of one and two active singularity in the sinh-Gordon theory \cite{Pozsgay2013}.
For the trace of the stress-energy tensor the conjecture is equivalent
to the excited state TBA equations for any state, similarly to the
Leclair-Mussardo series (\ref{eq:vacuum_LM}) which for the trace
of the stress-energy tensor is equivalent to the ground state TBA
(\ref{eq:vacuum_TBA},\ref{eq:vacuum_TBA_energy}) \cite{Leclair:1999ys}.
The proof of this equivalence is given in Section \ref{sec:Theta-TBA-LMEx}. 

The infrared limit of the formula also reproduces previously known
results for finite volume diagonal form factors which were obtained
in \cite{Pozsgay2008a,Pozsgay2008b}. In large volume the imaginary
parts of the active singularities tend to fixed values, which are
determined by the poles of the scattering matrix \cite{Dorey:1996re,Dorey:1997rb},
and the real parts $\{\vartheta_{j}\}=\{\mbox{Re}\,\bar{\theta}_{i}\}$
of the singularity positions can be interpreted as rapidities of on-shell
particles, where usually one particle is described by more than one
singularity positions, which all have the same real parts. The quantization
conditions reduce to the Bethe-Yang equations
\begin{equation}
m_{\alpha_{j}}\sinh\vartheta_{j}-i\sum_{k\neq j}\log S_{\alpha_{j}\alpha_{k}}(\vartheta_{j}-\vartheta_{k})=2\pi I_{j}\label{eq:Bethe-Yang}
\end{equation}
where the momentum quantum numbers $I_{j}$ can be obtained from the
quantum numbers $n_{i}$. The density of active singularities specified
in Definition \ref{def:density-of-active-sing} reduces to the usual
density of states in rapidity space, while the restricted density
in Definition \ref{def:restricted-density-of-active-sing} turns into
the restricted density used in the diagonal form factor formula in
\cite{Pozsgay2008b}.

In the same limit, the ``dressed'' diagonal form factors reduce
to connected diagonal form factors, and for theories where particles
are represented by a single active singularity, formula (\ref{eq:LMExcited_general})
reduces to the results obtained in \cite{Pozsgay2008a,Pozsgay2008b}
for the finite-volume diagonal matrix elements, which are valid up
to exponential corrections in the volume. 

In theories where a particle is represented by several active singularities,
the particle can be considered as a bound state of the active singularities.
In infinite volume this does not make any difference due to bootstrap
equations satisfied by the scattering matrix, but in finite volume
the composite nature of the particles gives exponential corrections,
which are exactly the $\mu$-term corrections to the form factor described
in \cite{Pozsgay2008}. However, while in \cite{Pozsgay2008} the
description of the particles as composite objects was still ambiguous,
the excited state TBA equation gives a clear prescription valid for
every value of the volume. 

In line with the usual terminology of finite volume corrections \cite{Luscher:1985dn,1991NuPhB.362..329K,2009NuPhB.807..625B},
the terms in (\ref{eq:LMExcited_general}) containing rapidity integration,
originating from either the quantization conditions (\ref{eq:excited_TBA_quantization_conditions})
or the dressed form factors (\ref{eq:dressed_FF}), give so-called
$F$-term corrections, which describe virtual particle loops winding
around the finite volume cylinder.

\section{Equivalence of the form factor series and the TBA for the trace of
the stress-energy tensor \label{sec:Theta-TBA-LMEx}}

In this section we present the equivalence of the conjectured form
factor series for excited states (\ref{eq:LMExcited_general}) and
the TBA equations for $\Theta$ the trace of the stress-energy tensor.
We proceed in three steps. First we explicitly evaluate the TBA prediction
for $\left\langle \Theta\right\rangle $, then recast it in a form
which can be matched with the dependence of (\ref{eq:LMExcited_general})
on the densities, and then prove that the rest of the formula matches
the dressed form factors of $\Theta$.

\subsection{$\left\langle \Theta\right\rangle $ from TBA}

As described in \cite{Zamolodchikov:1989cf} the expectation value
of the trace of the stress-energy tensor can be expressed in the following
way
\begin{eqnarray}
\left\langle \Theta\right\rangle _{L} & = & \left\langle \Theta\right\rangle _{\infty}+2\pi\left[\frac{E_{\mbox{TBA}}\left(L\right)}{L}+\frac{\mathrm{d}E_{\mbox{TBA}}\left(L\right)}{\mathrm{d}L}\right]
\end{eqnarray}
For an excited state with $N$ active singularities we obtain 
\begin{eqnarray}
E_{\mathrm{TBA}}\left(L\right) & = & \sum_{i=1}^{N}im_{\alpha_{i}}\eta_{i}\sinh\left(\bar{\theta}_{i}\right)-\sum_{\beta}\int\frac{\mathrm{d}\theta}{2\pi}m_{\beta}\sinh\left(\theta\right)\frac{\partial_{\theta}\varepsilon_{\beta}\left(\theta\right)}{1+e^{\varepsilon_{\beta}\left(\theta\right)}}\nonumber \\
\frac{\mathrm{d}E_{\mathrm{TBA}}\left(L\right)}{\mathrm{d}L} & = & \sum_{i=1}^{N}im_{\alpha_{i}}\eta_{i}\cosh\left(\bar{\theta}_{i}\right)\frac{d\bar{\theta}_{i}}{dL}+\sum_{\beta}\int\frac{\mathrm{d}\theta}{2\pi}m_{\beta}\cosh\left(\theta\right)\frac{\partial_{L}\varepsilon_{\beta}\left(\theta\right)}{1+e^{\varepsilon_{\beta}\left(\theta\right)}}
\end{eqnarray}
where we performed a partial integration in the energy expression.
The derivatives of the pseudo-energy satisfy the following linear
equations
\begin{eqnarray}
\partial_{\theta}\varepsilon_{\alpha}\left(\theta\right) & = & m_{\alpha}L\sinh\left(\theta\right)-\sum_{i=1}^{N}i\eta_{i}\varphi_{\alpha\alpha_{i}}\left(\theta-\bar{\theta}_{i}\right)+\sum_{\beta}\int\frac{\mathrm{d}\theta'}{2\pi}\varphi_{\alpha\beta}\left(\theta-\theta'\right)\frac{\partial_{\theta}\varepsilon_{\beta}\left(\theta'\right)}{1+e^{\varepsilon_{\beta}\left(\theta'\right)}}\nonumber \\
\partial_{L}\varepsilon_{\alpha}\left(\theta\right) & = & m_{\alpha}\cosh\left(\theta\right)+\sum_{i=1}^{N}i\eta_{i}\varphi_{\alpha\alpha_{i}}\left(\theta-\bar{\theta}_{i}\right)\frac{d\bar{\theta}_{i}}{dL}\label{eq:derivsofeps}\\
 &  & +\sum_{\beta}\int\frac{\mathrm{d}\theta'}{2\pi}\varphi_{\alpha\beta}\left(\theta-\theta'\right)\frac{\partial_{L}\varepsilon_{\beta}\left(\theta'\right)}{1+e^{\varepsilon_{\beta}\left(\theta'\right)}}\nonumber 
\end{eqnarray}
The linearity of the above equations can be exploited by introducing
new functions $f$ satisfying the following equations
\begin{eqnarray}
f_{s,\alpha}\left(\theta\right) & = & m_{\alpha}\sinh\left(\theta\right)+\sum_{\beta}\int\frac{\mathrm{d}\theta'}{2\pi}\varphi_{\alpha\beta}\left(\theta-\theta'\right)\frac{f_{s,\beta}\left(\theta'\right)}{1+e^{\varepsilon_{\beta}\left(\theta'\right)}}\nonumber \\
f_{c,\alpha}\left(\theta\right) & = & m_{\alpha}\cosh\left(\theta\right)+\sum_{\beta}\int\frac{\mathrm{d}\theta'}{2\pi}\varphi_{\alpha\beta}\left(\theta-\theta'\right)\frac{f_{c,\beta}\left(\theta'\right)}{1+e^{\varepsilon_{\beta}\left(\theta'\right)}}\nonumber \\
f_{i,\alpha}\left(\theta\right) & = & \varphi_{\alpha\alpha_{i}}\left(\theta-\bar{\theta}_{i}\right)+\sum_{\beta}\int\frac{\mathrm{d}\theta'}{2\pi}\varphi_{\alpha\beta}\left(\theta-\theta'\right)\frac{f_{i,\beta}\left(\theta'\right)}{1+e^{\varepsilon_{\beta}\left(\theta'\right)}}\label{eq:fsfcfi}
\end{eqnarray}
which can be used to express the derivatives as 
\begin{eqnarray}
\partial_{\theta}\varepsilon_{\alpha}\left(\theta\right) & = & Lf_{s,\alpha}\left(\theta\right)+\sum_{i=1}^{N}\left(-i\eta_{i}\right)f_{i,\alpha}\left(\theta\right)\nonumber \\
\partial_{L}\varepsilon_{\alpha}\left(\theta\right) & = & f_{c,\alpha}\left(\theta\right)+\sum_{i=1}^{N}\left(i\eta_{i}\frac{d\bar{\theta}_{i}}{dL}\right)f_{i,\alpha}\left(\theta\right)
\end{eqnarray}
Inserting these relation into (\ref{eq:derivsofeps})
\begin{eqnarray}
E_{\mathrm{TBA}}\left(L\right) & = & \sum_{i=1}^{N}im_{\alpha_{i}}\eta_{i}\sinh\left(\bar{\theta}_{i}\right)-\sum_{\beta}\int\frac{\mathrm{d}\theta}{2\pi}m_{\beta}\sinh\left(\theta\right)\frac{Lf_{s,\beta}\left(\theta\right)+\sum_{i=1}^{N}\left(-i\eta_{i}\right)f_{i,\beta}\left(\theta\right)}{1+e^{\varepsilon_{\beta}\left(\theta\right)}}\nonumber \\
\frac{\mathrm{d}E_{\mathrm{TBA}}\left(L\right)}{\mathrm{d}L} & = & \sum_{i=1}^{N}im_{\alpha_{i}}\eta_{i}\cosh\left(\bar{\theta}_{i}\right)\frac{d\bar{\theta}_{i}}{dL}\nonumber \\
 &  & +\sum_{\beta}\int\frac{\mathrm{d}\theta}{2\pi}m_{\beta}\cosh\left(\theta\right)\frac{f_{c,\beta}\left(\theta\right)+\sum_{i=1}^{N}\left(i\eta_{i}\frac{d\bar{\theta}_{i}}{dL}\right)f_{i,\beta}\left(\theta\right)}{1+e^{\varepsilon_{\beta}\left(\theta\right)}}
\end{eqnarray}
th expectation value $\left\langle \Theta\right\rangle $ takes the
form
\begin{eqnarray}
\frac{\left\langle \Theta\right\rangle }{2\pi} & = & \frac{\left\langle \Theta\right\rangle _{\infty}}{2\pi}+\sum_{i=1}^{N}\frac{i}{L}m_{\alpha_{i}}\eta_{i}\sinh\left(\bar{\theta}_{i}\right)-\sum_{\beta}\int\frac{\mathrm{d}\theta}{2\pi}m_{\beta}\sinh\left(\theta\right)\frac{f_{s,\beta}\left(\theta\right)+\sum_{i=1}^{N}\left(-\frac{i}{L}\eta_{i}\right)f_{i,\beta}\left(\theta\right)}{1+e^{\varepsilon_{\beta}\left(\theta\right)}}\nonumber \\
 &  & +\sum_{i=1}^{N}im_{\alpha_{i}}\eta_{i}\cosh\left(\bar{\theta}_{i}\right)\frac{d\bar{\theta}_{i}}{dL}+\sum_{\beta}\int\frac{\mathrm{d}\theta}{2\pi}m_{\beta}\cosh\left(\theta\right)\frac{f_{c,\beta}\left(\theta\right)+\sum_{i=1}^{N}\left(i\eta_{i}\frac{d\bar{\theta}_{i}}{dL}\right)f_{i,\beta}\left(\theta\right)}{1+e^{\varepsilon_{\beta}\left(\theta\right)}}\nonumber \\
 & = & \frac{\left\langle \Theta\right\rangle _{\infty}}{2\pi}+\sum_{\beta}\int\frac{\mathrm{d}\theta}{2\pi}\frac{m_{\beta}\cosh\left(\theta\right)f_{c,\beta}\left(\theta\right)-m_{\beta}\sinh\left(\theta\right)f_{s,\beta}\left(\theta\right)}{1+e^{\varepsilon_{\beta}\left(\theta\right)}}\nonumber \\
 &  & +\sum_{i=1}^{N}\frac{i}{L}\eta_{i}\left[m_{\alpha_{i}}\sinh\left(\bar{\theta}_{i}\right)+\sum_{\beta}\int\frac{\mathrm{d}\theta}{2\pi}m_{\beta}\sinh\left(\theta\right)\frac{f_{i,\beta}\left(\theta\right)}{1+e^{\varepsilon_{\beta}\left(\theta\right)}}\right]\nonumber \\
 &  & +\sum_{i=1}^{N}i\eta_{i}\frac{d\bar{\theta}_{i}}{dL}\left[m_{\alpha_{i}}\cosh\left(\bar{\theta}_{i}\right)+\sum_{\beta}\int\frac{\mathrm{d}\theta}{2\pi}m_{\beta}\cosh\left(\theta\right)\frac{f_{i,\beta}\left(\theta\right)}{1+e^{\varepsilon_{\beta}\left(\theta\right)}}\right]
\end{eqnarray}
Using that the derivatives of phase shift are even functions $\varphi_{\alpha\beta}\left(\theta\right)=\varphi_{\beta\alpha}\left(-\theta\right)$,
and the definition of $f_{i}$ and $f_{s,c}$ one can easily see that
\begin{eqnarray}
\sum_{\beta}\int\frac{\mathrm{d}\theta}{2\pi}m_{\beta}\sinh\left(\theta\right)\frac{f_{i,\beta}\left(\theta\right)}{1+e^{\varepsilon_{\beta}\left(\theta\right)}} & = & \sum_{\beta}\int\frac{\mathrm{d}\theta}{2\pi}\varphi_{\alpha_{i}\beta}\left(\bar{\theta}_{i}-\theta\right)\frac{f_{s,\beta}\left(\theta\right)}{1+e^{\varepsilon_{\beta}\left(\theta\right)}}\nonumber \\
\sum_{\beta}\int\frac{\mathrm{d}\theta}{2\pi}m_{\beta}\cosh\left(\theta\right)\frac{f_{i,\beta}\left(\theta\right)}{1+e^{\varepsilon_{\beta}\left(\theta\right)}} & = & \sum_{\beta}\int\frac{\mathrm{d}\theta}{2\pi}\varphi_{\alpha_{i}\beta}\left(\bar{\theta}_{i}-\theta\right)\frac{f_{c,\beta}\left(\theta\right)}{1+e^{\varepsilon_{\beta}\left(\theta\right)}}
\end{eqnarray}
and so $\left\langle \Theta\right\rangle $ simplifies to
\begin{eqnarray}
\frac{\left\langle \Theta\right\rangle }{2\pi} & = & \frac{\left\langle \Theta\right\rangle _{\infty}}{2\pi}+\sum_{\beta}\int\frac{\mathrm{d}\theta}{2\pi}\frac{m_{\beta}\cosh\left(\theta\right)f_{c,\beta}\left(\theta\right)-m_{\beta}\sinh\left(\theta\right)f_{s,\beta}\left(\theta\right)}{1+e^{\varepsilon_{\beta}\left(\theta\right)}}\nonumber \\
 &  & +\sum_{i=1}^{N}\frac{i}{L}\eta_{i}f_{s,\alpha_{i}}\left(\bar{\theta}_{i}\right)+\sum_{i=1}^{N}i\eta_{i}f_{c,\alpha_{i}}\left(\bar{\theta}_{i}\right)\frac{d\bar{\theta}_{i}}{dL}
\end{eqnarray}
The derivatives of the active singularity positions can be expressed
using the quantization conditions (\ref{eq:TBA_Qfunction})
\begin{eqnarray}
\frac{\mathrm{d}Q_{i}}{\mathrm{d}L} & = & \sum_{j}\frac{\mathrm{\partial}Q_{i}}{\partial\bar{\theta}_{j}}\frac{d\bar{\theta}_{j}}{dL}+\frac{\mathrm{\partial}Q_{i}}{\partial L}=0\nonumber \\
\frac{d\bar{\theta}_{i}}{dL} & = & -\sum_{j}\mathcal{K}_{ij}^{-1}\frac{\mathrm{\partial}Q_{j}}{\partial L}
\end{eqnarray}
where
\begin{eqnarray}
\mathcal{K}_{ij} & = & \frac{\mathrm{\partial}Q_{i}}{\partial\bar{\theta}_{j}}=i\eta_{i}\frac{\mathrm{\partial}\varepsilon_{\alpha_{i}}\left(\bar{\theta}_{i}|\bar{\theta}_{1},\dots,\bar{\theta}_{N}\right)}{\partial\bar{\theta}_{j}}\nonumber \\
 & = & i\eta_{i}\begin{cases}
Lf_{s,\alpha_{i}}\left(\bar{\theta}_{i}\right)+\sum_{k\neq i}\left(-i\eta_{k}\right)f_{k,\alpha_{i}}\left(\bar{\theta}_{i}\right) & i=j\\
\left(i\eta_{j}\right)f_{j,\alpha_{i}}\left(\bar{\theta}_{i}\right) & i\neq j
\end{cases}\label{eq:Density_general}
\end{eqnarray}
Introducing the following combinations
\begin{eqnarray}
\mathcal{N}_{i} & = & i\eta_{i}f_{s,\alpha_{i}}\left(\bar{\theta}_{i}\right)\nonumber \\
\mathcal{N}_{\varphi,ij} & = & \eta_{i}\eta_{j}f_{j,\alpha_{i}}\left(\bar{\theta}_{i}\right)\label{eq:defN_I_N_phi}
\end{eqnarray}
where $\mathcal{N}_{\varphi,ij}=\mathcal{N}_{\varphi,ji}$ , $\mathcal{K}$
can be rewritten as
\begin{eqnarray}
\mathcal{K}_{ij} & = & \begin{cases}
L\mathcal{N}_{i}+\sum_{k\neq i}\mathcal{N}_{\varphi,ik} & i=j\\
-\mathcal{N}_{\varphi,ik} & i\neq j
\end{cases}
\end{eqnarray}
The explicit volume derivative of the quantization condition is
\begin{eqnarray}
\frac{\mathrm{\partial}Q_{i}}{\partial L} & = & i\eta_{i}\frac{\mathrm{\partial}\varepsilon_{\alpha_{i}}\left(\bar{\theta}_{i}|\bar{\theta}_{1},\dots,\bar{\theta}_{N}\right)}{\partial L}=i\eta_{i}f_{c,\alpha_{i}}\left(\bar{\theta}_{i}\right)
\end{eqnarray}
Introducing
\begin{eqnarray}
\mathcal{M}_{i} & = & i\eta_{i}f_{c,\alpha_{i}}\left(\bar{\theta}_{i}\right)\label{eq:defM_i}
\end{eqnarray}
the derivative of the singularity position takes the following form
\begin{eqnarray}
\frac{d\bar{\theta}_{i}}{dL} & = & -\sum_{j}\mathcal{K}_{ij}^{-1}\mathcal{M}_{j}
\end{eqnarray}
such as $\left\langle \Theta\right\rangle $

\begin{eqnarray}
\frac{\left\langle \Theta\right\rangle }{2\pi} & = & \frac{\left\langle \Theta\right\rangle _{\infty}}{2\pi}+\sum_{\beta}\int\frac{\mathrm{d}\theta}{2\pi}\frac{m_{\beta}\cosh\left(\theta\right)f_{c,\beta}\left(\theta\right)-m_{\beta}\sinh\left(\theta\right)f_{s,\beta}\left(\theta\right)}{1+e^{\varepsilon_{\beta}\left(\theta\right)}}\nonumber \\
 &  & +\sum_{i=1}^{N}\frac{\mathcal{N}_{i}}{L}-\sum_{i,j=1}^{N}\mathcal{M}_{i}\mathcal{K}_{ij}^{-1}\mathcal{M}_{j}\label{eq:Theta_almost_fin}
\end{eqnarray}
This is our final form for the TBA result for $\left\langle \Theta\right\rangle $.

\subsection{Isolating the singularity density terms}

To see the equivalence of $\left\langle \Theta\right\rangle $ to
the form factor series (\ref{eq:LMExcited_general}) the terms containing
$\mathcal{N}_{i}$ and $\mathcal{M}_{i}$ need to be rearranged in
order to match the structure of the singularity density terms in (\ref{eq:LMExcited_general}).

Let's start with the term 
\begin{equation}
\sum_{i,j=1}^{N}\mathcal{M}_{i}\mathcal{K}_{ij}^{-1}\mathcal{M}_{j}
\end{equation}
The inverse of $\mathcal{K}$ can be expressed by its co-factor matrix
$\mathcal{C}$ 
\begin{eqnarray}
\mathcal{K}_{ij}^{-1} & = & \frac{\mathcal{C}_{ji}}{\det\mathcal{K}}
\end{eqnarray}
The diagonal elements of the co-factor matrix are just the principal
minors of $\mathcal{K}$: 
\begin{eqnarray}
\mathcal{C}_{ii} & = & \det\mathcal{K}\left(\{i\}\right)\label{eq:cofactor_diag}
\end{eqnarray}
where $\mathcal{K}\left(I\right)$ denotes the matrix obtained by
omitting from $\mathcal{K}$ the rows and columns that are indexed
by the set $I$. The non-diagonal elements of the co-factor matrix
can be expressed with principal minors and sequences of the elements
of $\mathcal{K}$ \cite{John-S.-Maybee:1989aa}

\begin{eqnarray}
\mathcal{C}_{ji} & = & \sum_{n=0}^{N-2}\sum_{\left\{ \alpha\right\} }\left(-1\right)^{n+1}\mathcal{K}_{i\alpha_{1}}\mathcal{K}_{\alpha_{1}\alpha_{2}}\dots\mathcal{K}_{\alpha_{n}j}\det\mathcal{\mathcal{K}}\left(\{j,i,\alpha_{1},\dots,\alpha_{n}\}\right)\nonumber \\
 & = & \sum_{n=0}^{N-2}\sum_{\left\{ \alpha\right\} }\mathcal{N}_{\varphi,i\alpha_{1}}\mathcal{N}_{\varphi,\alpha_{1}\alpha_{2}}\dots\mathcal{N}_{,\varphi\alpha_{n}j}\det\mathcal{\mathcal{K}}\left(\{j,i,\alpha_{1},\dots,\alpha_{n}\}\right)\label{eq:cofactor_offdiag}
\end{eqnarray}
where $\left\{ \alpha\right\} \subset\left\{ 1,\dots,N\right\} \setminus\left\{ i,j\right\} $
. With the help of (\ref{eq:cofactor_diag}) and (\ref{eq:cofactor_offdiag})
one can write
\begin{eqnarray}
\sum_{i,j=1}^{N}\mathcal{M}_{i}\mathcal{K}_{ij}^{-1}\mathcal{M}_{j} & = & \sum_{i}\frac{\det\mathcal{K}\left(i\right)}{\det\mathcal{K}}\mathcal{M}_{i}\mathcal{M}_{i}+\sum_{i\neq j}\sum_{n=0}^{N-2}\sum_{\left\{ \alpha\right\} }\frac{\det\mathcal{\mathcal{K}}\left(j,i,\left\{ \alpha\right\} \right)}{\det\mathcal{K}}\nonumber \\
 &  & \times\mathcal{M}_{i}\mathcal{M}_{j}\mathcal{N}_{\varphi,i\alpha_{1}}\mathcal{N}_{\varphi,\alpha_{1}\alpha_{2}}\dots\mathcal{N}_{,\varphi\alpha_{n}j}\label{eq:MKM_goodform}
\end{eqnarray}
Now we turn to rearranging the term 
\begin{equation}
\sum_{i=1}^{N}\frac{\mathcal{N}_{i}}{L}
\end{equation}
in a similar manner. For this we need the following theorem:
\begin{thm}
\label{thm:DetExpandforNi}If the $N\times N$ matrix $\mathcal{K}^{(N)}$
has the form
\begin{eqnarray}
\mathcal{K}_{ij}^{(N)} & = & \begin{cases}
L\mathcal{N}_{i}+\sum_{k\neq i}\mathcal{N}_{\varphi,ik} & i=j\\
-\mathcal{N}_{\varphi,ik} & i\neq j
\end{cases}
\end{eqnarray}
its determinant can be expanded as
\begin{eqnarray}
\det\mathcal{K}^{(N)} & = & L\mathcal{N}_{i}\det\mathcal{K}^{(N)}\left(\{i\}\right)\\
 &  & +\sum_{n=1}^{N-1}\sum_{\left\{ \alpha\right\} }\mathcal{N}_{\varphi,i\alpha_{1}}\mathcal{N}_{\varphi,\alpha_{1}\alpha_{2}}\dots\mathcal{N}_{\varphi,\alpha_{n-1}\alpha_{n}}L\mathcal{N}_{\alpha_{n}}\det\mathcal{\mathcal{K}}^{(N)}\left(\{i,\alpha_{1},\dots,\alpha_{n}\}\right)\nonumber 
\end{eqnarray}
where $i$ is any chosen row, $\left\{ \alpha\right\} \subset\left\{ 1,\dots,N\right\} \setminus\left\{ i\right\} $
and $\mathcal{K}^{(N)}\left(I\right)$ is the submatrix of $\mathcal{K}^{(N)}$
as defined before.\end{thm}
\begin{proof}
Up to $N=3$ it is easy to check the statement by direct evaluation.
For $N>3$ we proceed by induction. Let us suppose the theorem is
valid for $N-1$
\begin{eqnarray}
\det\mathcal{K}^{(N-1)} & = & L\mathcal{N}_{i}\det\mathcal{K}^{(N-1)}\left(\{i\}\right)+\sum_{n=1}^{N-2}\sum_{\left\{ \alpha\right\} }\Big\{\mathcal{N}_{\varphi,i\alpha_{1}}\mathcal{N}_{\varphi,\alpha_{1}\alpha_{2}}\dots\mathcal{N}_{\varphi,\alpha_{n-1}\alpha_{n}}L\mathcal{N}_{\alpha_{n}}\nonumber \\
 &  & \quad\times\det\mathcal{\mathcal{K}}^{(N-1)}\left(\{i,\alpha_{1},\dots,\alpha_{n}\}\right)\Big\}\label{eq:detKN-1}
\end{eqnarray}
The determinant for the matrix $\mathcal{K}^{(N)}$ of size $N$ can
be expanded by its row $j$
\begin{eqnarray}
\det\mathcal{K}^{(N)} & = & \mathcal{K}_{jj}^{(N)}\mathcal{C}_{jj}^{(N)}+\sum_{i\neq j}\mathcal{K}_{ji}^{(N)}\mathcal{C}_{ji}^{(N)}
\end{eqnarray}
where $\mathcal{C}^{(N)}$ is the co-factor matrix of $\mathcal{K}^{(N)}$.
Using (\ref{eq:cofactor_diag}) and (\ref{eq:cofactor_offdiag}) leads
to 
\begin{eqnarray}
\det\mathcal{K}^{(N)} & = & \mathcal{K}_{jj}^{(N)}\det\mathcal{K}^{(N)}\left(\{j\}\right)+\sum_{i\neq j}\mathcal{K}_{ji}^{(N)}\sum_{n=0}^{N-2}\sum_{\left\{ \alpha\right\} }\Big\{\left(-1\right)^{n+1}\mathcal{K}_{i\alpha_{1}}\mathcal{K}_{\alpha_{1}\alpha_{2}}\dots\mathcal{K}_{\alpha_{n}j}\nonumber \\
 &  & \quad\times\det\mathcal{\mathcal{K}}^{(N)}\left(\{j,i,\alpha_{1},\dots,\alpha_{n}\}\right)\Big\}\nonumber \\
 & = & \left(L\mathcal{N}_{j}+\sum_{i\neq j}\mathcal{N}_{\varphi,ji}\right)\det\mathcal{K}^{(N)}\left(\left\{ j\right\} \right)-\sum_{i\neq j}\Big\{\mathcal{N}_{\varphi,ji}\nonumber \\
 &  & \quad\times\sum_{n=0}^{N-2}\sum_{\left\{ \alpha\right\} }\mathcal{N}_{\varphi,i\alpha_{1}}\mathcal{N}_{\varphi,\alpha_{1}\alpha_{2}}\dots\mathcal{N}_{\varphi,\alpha_{n}j}\det\mathcal{\mathcal{K}}^{(N)}\left(\{j,i,\alpha_{1},\dots,\alpha_{n}\}\right)\Big\}\label{eq:detKNbasic}
\end{eqnarray}
Now $\mathcal{K}^{(N)}\left(\{j\}\right)$ can be related to $\mathcal{K}^{(N-1)}$
by observing that their off-diagonal elements are the same, while
$\mathcal{K}_{ii}^{(N)}=\mathcal{K}_{ii}^{(N-1)}+\mathcal{N}_{\varphi,ij}$.
Implementing this by shifting $L\mathcal{N}_{i}\to L\mathcal{N}_{i}+\mathcal{N}_{\varphi,ij}$
in (\ref{eq:detKN-1}) one obtains 
\begin{eqnarray}
\det\mathcal{K}^{N}\left(\{j\}\right) & = & \left.\det\mathcal{K}^{(N-1)}\right|_{L\mathcal{N}_{i}\to L\mathcal{N}_{i}+\mathcal{N}_{\varphi,ij}}\nonumber \\
 & = & \left(L\mathcal{N}_{i}+\mathcal{N}_{\varphi,ij}\right)\det\mathcal{K}^{(N)}\left(\{j,i\}\right)+\sum_{n=1}^{N-2}\sum_{\left\{ \alpha\right\} }\Big\{\mathcal{N}_{\varphi,i\alpha_{1}}\mathcal{N}_{\varphi,\alpha_{1}\alpha_{2}}\dots\mathcal{N}_{\varphi,\alpha_{n-1}\alpha_{n}}\nonumber \\
 &  & \quad\times\left(L\mathcal{N}_{\alpha_{n}}+\mathcal{N}_{\varphi,\alpha_{n}j}\right)\det\mathcal{\mathcal{K}}^{(N)}\left(\{j,i,\alpha_{1},\dots,\alpha_{n}\}\right)\Big\}
\end{eqnarray}
and inserting this back to (\ref{eq:detKNbasic}) gives 
\begin{eqnarray}
\det\mathcal{K}^{(N)} & = & L\mathcal{N}_{j}\det\mathcal{K}^{(N)}\left(\{j\}\right)+\sum_{i\neq j}\mathcal{N}_{\varphi,ji}\left(L\mathcal{N}_{i}+\mathcal{N}_{\varphi,ij}\right)\det\mathcal{K}^{(N)}\left(\{j,i\}\right)\nonumber \\
 &  & +\sum_{i\neq j}\mathcal{N}_{\varphi,ji}\sum_{n=1}^{N-2}\sum_{\left\{ \alpha\right\} }\Big\{\mathcal{N}_{\varphi,i\alpha_{1}}\mathcal{N}_{\varphi,\alpha_{1}\alpha_{2}}\dots\mathcal{N}_{\varphi,\alpha_{n-1}\alpha_{n}}\nonumber \\
 &  & \quad\times\left(L\mathcal{N}_{\alpha_{n}}+\mathcal{N}_{\varphi,\alpha_{n}j}\right)\det\mathcal{\mathcal{K}}^{(N)}\left(\{j,i,\alpha_{1},\dots,\alpha_{n}\}\right)\Big\}\nonumber \\
 &  & -\sum_{i\neq j}\mathcal{N}_{\varphi,ji}\sum_{n=0}^{N-2}\sum_{\left\{ \alpha\right\} }\mathcal{N}_{\varphi,i\alpha_{1}}\mathcal{N}_{\varphi,\alpha_{1}\alpha_{2}}\dots\mathcal{N}_{\varphi,\alpha_{n}j}\det\mathcal{\mathcal{K}}^{(N)}\left(\{i,j,\alpha_{1},\dots,\alpha_{n}\}\right)\nonumber \\
 & = & L\mathcal{N}_{j}\det\mathcal{K}^{(N)}\left(\{j\}\right)+\nonumber \\
 &  & +\sum_{n=1}^{N-1}\sum_{\left\{ \alpha\right\} }\mathcal{N}_{\varphi,j\alpha_{1}}\mathcal{N}_{\varphi,\alpha_{1}\alpha_{2}}\dots\mathcal{N}_{\varphi,\alpha_{n-1}\alpha_{n}}L\mathcal{N}_{\alpha_{n}}\det\mathcal{\mathcal{K}}^{(N)}\left(\{j,\alpha_{1},\dots,\alpha_{n}\}\right)
\end{eqnarray}
which is just the statement we wanted to prove. \emph{Q.e.d.}
\end{proof}
Using the above theorem we can rewrite 
\begin{eqnarray}
\sum_{i}\frac{\mathcal{N}_{i}}{L} & = & \sum_{i}\frac{\mathcal{N}_{i}}{L}\frac{\det\mathcal{K}}{\det\mathcal{K}}=\sum_{i}\mathcal{N}_{i}\mathcal{N}_{i}\frac{\det\mathcal{K}\left(\{i\}\right)}{\det\mathcal{K}}+\sum_{i\neq j}\sum_{n=0}^{N-2}\sum_{\left\{ \alpha\right\} }\frac{\det\mathcal{\mathcal{K}}^{N}\left(\{i,j,\alpha_{1},\dots,\alpha_{n}\}\right)}{\det\mathcal{K}}\nonumber \\
 &  & \times\mathcal{N}_{i}\mathcal{N}_{j}\mathcal{N}_{\varphi,i\alpha_{1}}\mathcal{N}_{\varphi,\alpha_{1}\alpha_{2}}\dots\mathcal{N}_{\varphi,\alpha_{n}j}\label{eq:N/L_goodform}
\end{eqnarray}
which has the same structure as (\ref{eq:MKM_goodform}). Substituting
the definitions of $\mathcal{N}$ (\ref{eq:defN_I_N_phi}) and $\mathcal{M}$
(\ref{eq:defM_i}) into (\ref{eq:MKM_goodform}) and (\ref{eq:N/L_goodform})
we can see that every $\eta_{i}$ factor appears twice and so drops
out. Therefore the expression for $\left\langle \Theta\right\rangle $
(\ref{eq:Theta_almost_fin}) simplifies to
\begin{eqnarray}
\frac{\left\langle \Theta\right\rangle }{2\pi} & = & \frac{\left\langle \Theta\right\rangle _{\infty}}{2\pi}+\sum_{\beta}\int\frac{\mathrm{d}\theta}{2\pi}\frac{m_{\beta}\cosh\left(\theta\right)f_{c,\beta}\left(\theta\right)-m_{\beta}\sinh\left(\theta\right)f_{s,\beta}\left(\theta\right)}{1+e^{\varepsilon_{\beta}\left(\theta\right)}}\nonumber \\
 &  & +\sum_{i}\frac{\det\mathcal{K}\left(\{i\}\right)}{\det\mathcal{K}}\left[f_{c,\alpha_{i}}\left(\bar{\theta}_{i}\right)f_{c,\alpha_{i}}\left(\bar{\theta}_{i}\right)-f_{s,\alpha_{i}}\left(\bar{\theta}_{i}\right)f_{s,\alpha_{i}}\left(\bar{\theta}_{i}\right)\right]\nonumber \\
 &  & +\sum_{i\neq j}\sum_{n=0}^{N-2}\sum_{\left\{ \alpha\right\} }\frac{\det\mathcal{\mathcal{K}}\left(\{i,j,\alpha_{1},\dots,\alpha_{n}\}\right)}{\det\mathcal{K}}\left[f_{c,\alpha_{i}}\left(\bar{\theta}_{i}\right)f_{c,\alpha_{j}}\left(\bar{\theta}_{j}\right)-f_{s,\alpha_{i}}\left(\bar{\theta}_{i}\right)f_{s,\alpha_{j}}\left(\bar{\theta}_{j}\right)\right]\nonumber \\
 &  & \times f_{\alpha_{1},\alpha_{i}}\left(\bar{\theta}_{i}\right)f_{\alpha_{2},\alpha_{1}}\left(\bar{\theta}_{1}\right)\dots f_{j,\alpha_{n}}\left(\bar{\theta}_{n}\right)
\end{eqnarray}
The determinant ratios are exactly the density factors in (\ref{eq:LMExcited_general});
what remains to be shown is that the other terms reproduce the dressed
form factors of $\Theta$.

\subsection{Dressed form factors of $\Theta$}
\begin{thm}
\label{thm:Dressed_FF_theorem_basic}In the absence of active singularities
of the TBA equations, the dressed form factors of $\Theta$ are given
by
\begin{eqnarray}
\mathcal{D}_{\varepsilon}^{\Theta} & = & \sum_{n_{1},\dots,n_{k}=0}^{\infty}\frac{1}{\prod_{i}n_{i}!}\int_{-\infty}^{\infty}\prod_{j=1}^{\tilde{N}}\frac{\mathrm{d}\theta_{j}}{2\pi\left[1+e^{\varepsilon_{\beta_{j}}\left(\theta_{j}\right)}\right]}F_{2n_{1},\dots,2n_{k},c}^{\Theta}\left(\theta_{1},\dots,\theta_{\tilde{N}}\right)\label{eq:FF_Conn_Theta_theroem_basic}
\end{eqnarray}
which is equal to
\begin{eqnarray}
\mathcal{D}_{\varepsilon}^{\Theta} & = & \left\langle \Theta\right\rangle _{\infty}+2\pi\sum_{\beta}\int\frac{\mathrm{d}\theta}{2\pi}\frac{m_{\beta}\cosh\left(\theta\right)f_{c,\beta}\left(\theta\right)-m_{\beta}\sinh\left(\theta\right)f_{s,\beta}\left(\theta\right)}{1+e^{\varepsilon_{\beta}\left(\theta\right)}}\label{eq:theor_ff_basic}
\end{eqnarray}
\end{thm}
\begin{proof}
The connected diagonal form factors of $\Theta$ are given by \cite{Leclair:1999ys}
\begin{eqnarray*}
F_{2n,c}^{\Theta}\left(\theta_{1},\dots,\theta_{n}\right) & = & 2\pi\varphi_{12}\varphi_{23}\dots\varphi_{n-1,n}m_{\beta{}_{1}}m_{\beta_{n}}\cosh\left(\theta_{1n}\right)+\mathrm{permutations}
\end{eqnarray*}
where $\theta_{ij}=\theta_{i}-\theta_{j}$, $\beta_{i}$ denotes the
species of the $i$th particle and $\varphi_{ij}$ is a short-hand
for $\varphi_{\beta_{i}\beta_{j}}(\theta_{ij})$. (\ref{eq:FF_Conn_Theta_theroem_basic})
is symmetric under re-ordering particles of the same species which
results in a combinatorial factor $\prod_{i}n_{i}!$ canceling the
denominators in front of the integrals. To take into account the rest
of the permutations we can rewrite $\mathcal{D}_{\varepsilon}^{\Theta}$
like
\begin{eqnarray}
\mathcal{D}_{\varepsilon}^{\Theta} & = & 2\pi\sum_{n=0}^{\infty}\sum_{\beta_{1},\dots,\beta_{n}}\int_{-\infty}^{\infty}\prod_{j=1}^{n}\frac{\mathrm{d}\theta_{j}}{2\pi\left[1+e^{\varepsilon_{\beta_{j}}\left(\theta_{j}\right)}\right]}\varphi_{12}\varphi_{23}\dots\varphi_{n-1,n}m_{\beta{}_{1}}m_{\beta_{n}}\cosh\left(\theta_{1n}\right)\nonumber \\
 & = & \sum_{n=0}^{\infty}\mathcal{D}_{\varepsilon,n}^{\Theta}
\end{eqnarray}
Following \cite{Leclair:1999ys} every $\mathcal{D}_{\varepsilon,n}^{\Theta}$
can be graphically represented as seen in Fig. \ref{fig:graphicalLM_basic}
where every node represents a particle with a given rapidity and species,
including the integration 
\[
\int\frac{\mathrm{d}\theta_{i}}{2\pi\left[1+e^{\varepsilon_{\beta_{i}}\left(\theta_{i}\right)}\right]}
\]
and the first and last node is multiplied by its mass. Every horizontal
line represents a factor $\varphi_{ij}$ and the dashed line represents
the factor $\cosh\theta_{1n}$. The whole graph is multiplied by $2\pi$
to account for the normalization of the operator $\Theta$, and summed
over every possible type configuration for the nodes. The empty graph
(with zero node) represents $\mathcal{D}_{\varepsilon,0}^{\Theta}=\left\langle \Theta\right\rangle _{\infty}$.
Using hyperbolic addition formulas for the $\cosh\theta_{1n}$ terms
every graph can be represented as difference of two chains where the
two end nodes instead of being connected by dashed line, are multiplied
by $\cosh$ or $\sinh$ of the rapidity at the given node as shown
in Fig. \ref{fig:graphicalLM_basic_expand}.

\begin{figure}
\centering{}\subfloat[\label{fig:graphicalLM_basic} $\mathcal{D}_{\varepsilon,n}^{\Theta}$.]{\centering
\def\svgwidth{0.4\columnwidth}
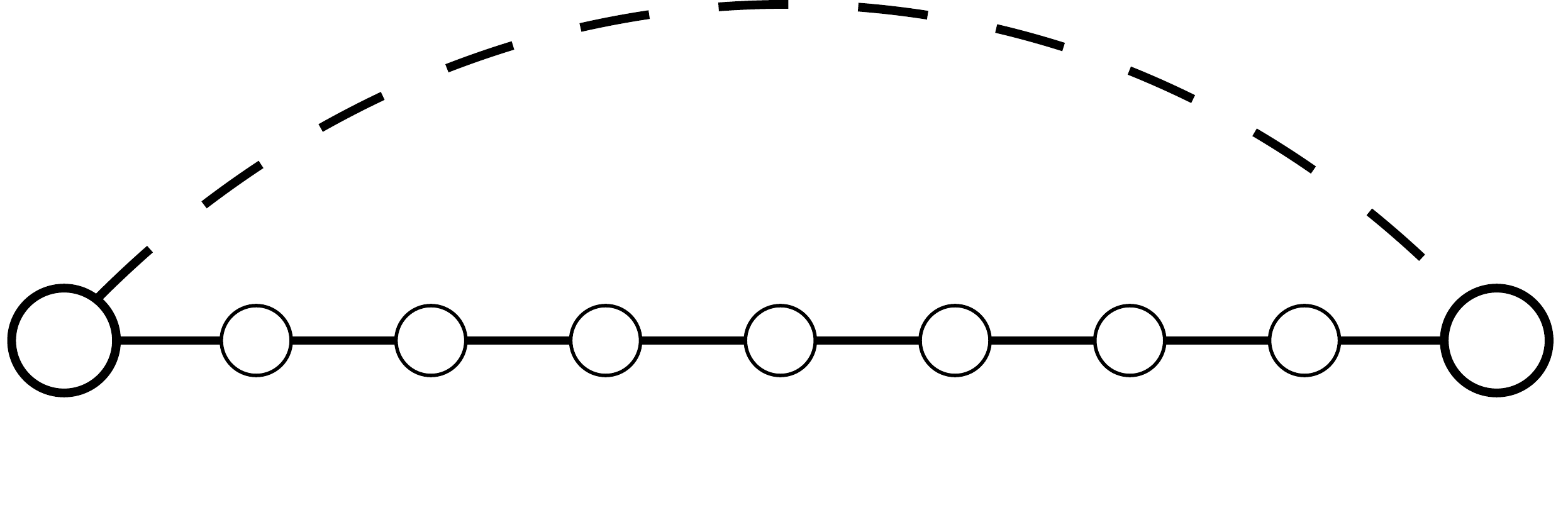}\hfill{}\subfloat[\label{fig:graphicalLM_basic_expand} $\mathcal{D}_{\varepsilon,n}^{\Theta}$
expanded.]{\centering
\def\svgwidth{0.4\columnwidth}
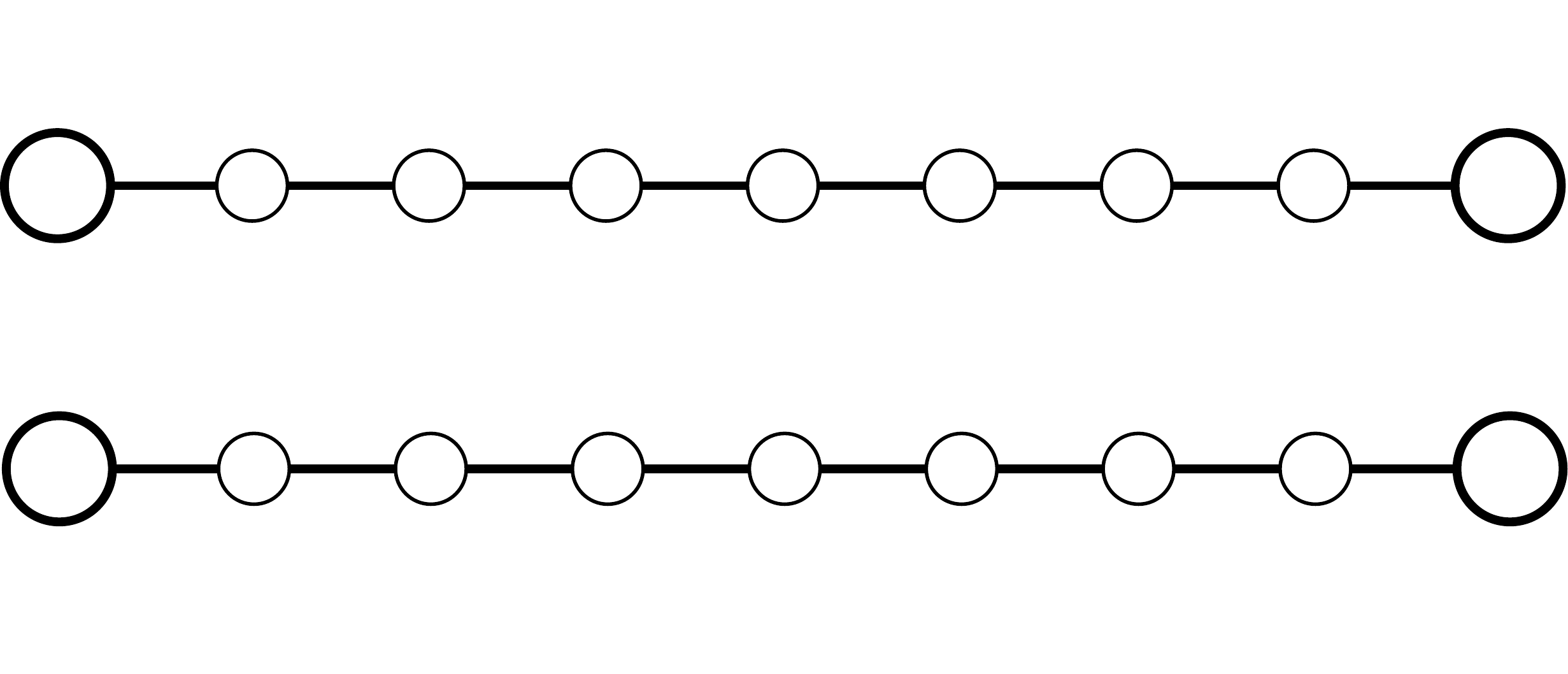

}\smallskip{}
\subfloat[\label{fig:graphicalLM_basic_fc_fs} $K_{n,\beta}\left(\theta\right)$
and $J_{n,\beta}\left(\theta\right)$.]{\centering
\def\svgwidth{0.4\columnwidth}
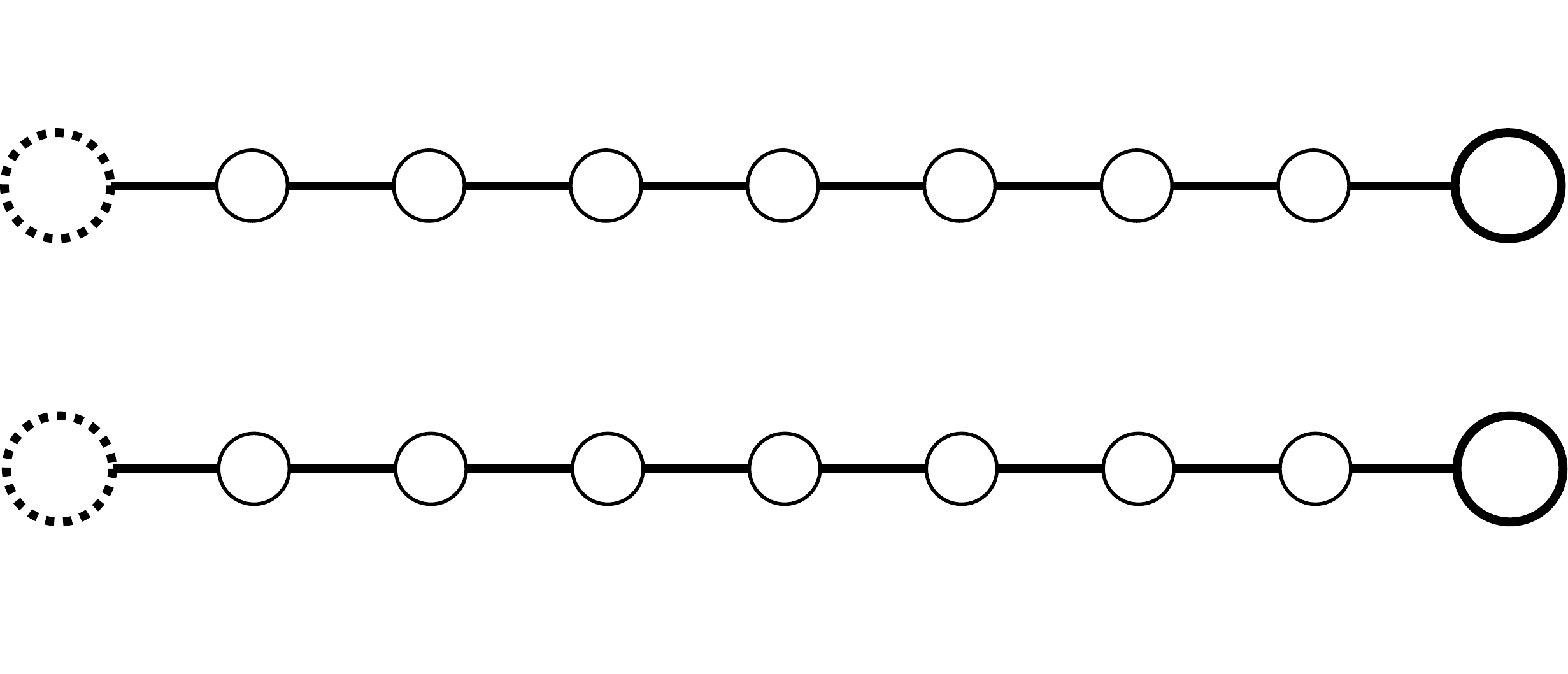

}\protect\caption{\label{fig:graphicalLM_basic_full}Graphical representation of $\mathcal{D}_{\varepsilon,n}^{\Theta}$,
$K_{n,\beta}\left(\theta\right)$ and $J_{n,\beta}\left(\theta\right)$.}
\end{figure}
Since the functions $f_{c}$ and $f_{s}$ in (\ref{eq:theor_ff_basic})
satisfy the self-consistent equations (\ref{eq:fsfcfi}), it is convenient
to expand them in the following way
\begin{eqnarray}
f_{c,\beta}\left(\theta\right) & = & \sum_{n=0}^{\infty}K_{n,\beta}\left(\theta\right)\nonumber \\
f_{c,\beta}\left(\theta\right) & = & \sum_{n=0}^{\infty}J_{n,\beta}\left(\theta\right)
\end{eqnarray}
where
\begin{eqnarray}
K_{0,\beta}\left(\theta\right) & = & m_{\beta}\cosh\left(\theta\right)\nonumber \\
J_{0,\beta}\left(\theta\right) & = & m_{\beta}\sinh\left(\theta\right)\nonumber \\
K_{n,\beta}\left(\theta\right) & = & \sum_{\beta_{1},\dots,\beta_{n}}\int\prod_{i}\frac{\mathrm{d}\theta_{i}}{2\pi\left[1+e^{\varepsilon_{\beta_{i}}\left(\theta_{i}\right)}\right]}\varphi_{\beta\beta_{1}}\dots\varphi_{\beta_{n-1}\beta_{n}}m_{\beta_{n}}\cosh\left(\theta_{n}\right)\nonumber \\
J_{n,\beta}\left(\theta\right) & = & \sum_{\beta_{1},\dots,\beta_{n}}\int\prod_{i}\frac{\mathrm{d}\theta_{i}}{2\pi\left[1+e^{\varepsilon_{\beta_{i}}\left(\theta_{i}\right)}\right]}\varphi_{\beta\beta_{1}}\dots\varphi_{\beta_{n-1}\beta_{n}}m_{\beta_{n}}\sinh\left(\theta_{n}\right)
\end{eqnarray}
The graphical representation of $K_{n,\beta}\left(\theta\right)$
and $J_{n,\beta}\left(\theta\right)$ can be seen in Fig. \ref{fig:graphicalLM_basic_fc_fs};
the dashed node indicates that the corresponding rapidity integral
and the filling fraction belonging to that node is not included in
the contribution. Comparing to Fig. \ref{fig:graphicalLM_basic_expand}
it is clear that $K_{n,\beta}\left(\theta\right)$ and $J_{n,\beta}\left(\theta\right)$
describe the contribution of the chain between one of the end nodes
and a dashed node with type $\beta$ which is $n$ steps away from
the end node. Multiplying $K_{n,\beta}\left(\theta\right)$ and $J_{n,\beta}\left(\theta\right)$
with $2\pi\sum_{\beta}\int\frac{\mathrm{d}\theta}{2\pi}\frac{m_{\beta}\cosh\left(\theta\right)}{1+e^{\varepsilon_{\beta}\left(\theta\right)}}$
and $-2\pi\sum_{\beta}\int\frac{\mathrm{d}\theta}{2\pi}\frac{m_{\beta}\sinh\left(\theta\right)f_{s,\beta}\left(\theta\right)}{1+e^{\varepsilon_{\beta}\left(\theta\right)}}$
as in (\ref{eq:theor_ff_basic}) closes the chains and they become
identical to the ones in Fig. \ref{fig:graphicalLM_basic_expand}
with length $n+1$, i.e. equal to $\mathcal{D}_{\varepsilon,n+1}^{\Theta}$.
The sum for $n$ in $f_{c}$ and $f_{s}$ then generates all the contributions
in $\mathcal{D}_{\varepsilon}^{\Theta}$. \emph{Q.e.d.}\end{proof}
\begin{thm}
\label{thm:Dressed_FF_theorem_Ex1}The dressed form factor of $\Theta$
with one active singularity $\bar{\theta}_{i}$ with type $\alpha_{i}$
is
\begin{eqnarray}
\mathcal{D_{\varepsilon}^{O}}\left(\bar{\theta}_{i}\right) & = & \sum_{n_{1},\dots,n_{k}=0}^{\infty}\frac{1}{\prod_{i}n_{i}!}\int_{-\infty}^{\infty}\prod_{j=1}^{\tilde{N}}\frac{\mathrm{d}\theta_{j}}{2\pi\left[1+e^{\varepsilon_{\beta_{j}}\left(\theta_{j}\right)}\right]}F_{2n_{1},\dots,2n_{k},c}^{\Theta}\left(\bar{\theta}_{i},\theta_{1},\dots,\theta_{\tilde{N}}\right)\label{eq:FF_conn_Theta_theorem_Ex1}
\end{eqnarray}
which is equal to
\begin{eqnarray}
\mathcal{D_{\varepsilon}^{O}}\left(\bar{\theta}_{i}\right) & = & f_{c,\alpha_{i}}\left(\bar{\theta}_{i}\right)f_{c,\alpha_{i}}\left(\bar{\theta}_{i}\right)-f_{s,\alpha_{i}}\left(\bar{\theta}_{i}\right)f_{s,\alpha_{i}}\left(\bar{\theta}_{i}\right)\label{eq:theor_ff_Ex1}
\end{eqnarray}
\end{thm}
\begin{proof}
\begin{figure}
\centering{}\subfloat[\label{fig:graphicalLM_Ex1} $\mathcal{D}_{\varepsilon,n\, m}^{\Theta}\left(\bar{\theta}_{i}\right)$.]{\centering
\def\svgwidth{0.4\columnwidth}
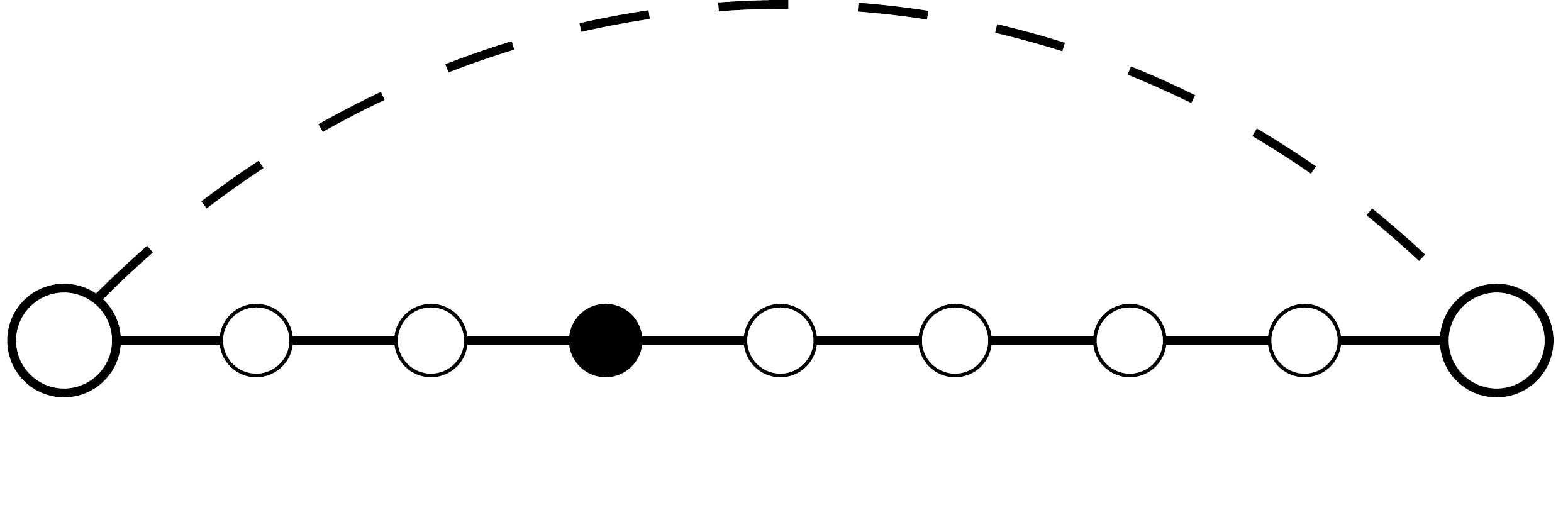}\hfill{}\subfloat[\label{fig:graphicalLM_Ex1_expand}$\mathcal{D}_{\varepsilon,n\, m}^{\Theta}\left(\bar{\theta}_{i}\right)$
expanded.]{\centering
\def\svgwidth{0.4\columnwidth}
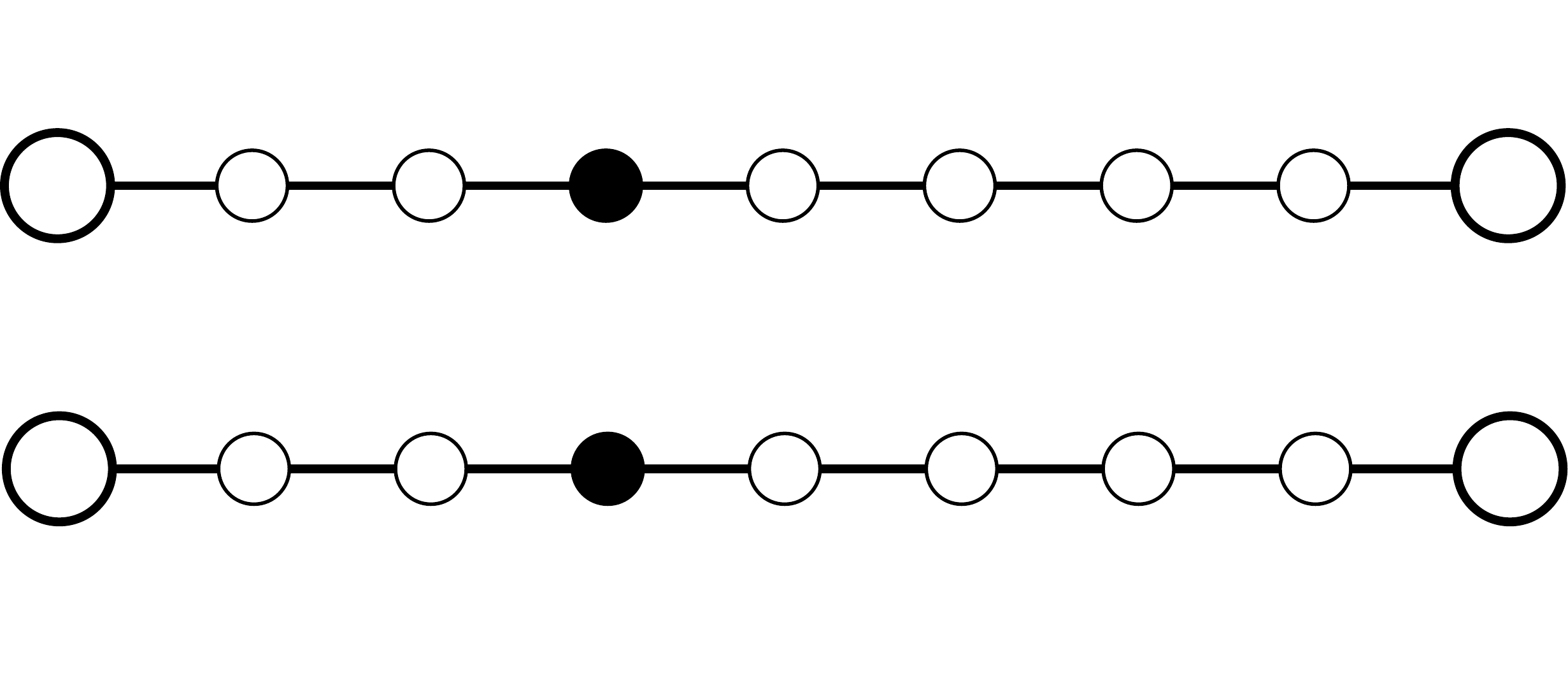

}\smallskip{}
\subfloat[\label{fig:graphicalLM_bEx1_fc_fs}$K_{n,\alpha_{i}}\left(\bar{\theta}_{i}\right)$
and $J_{n,\alpha_{i}}\left(\bar{\theta}_{i}\right)$.]{\centering
\def\svgwidth{0.4\columnwidth}
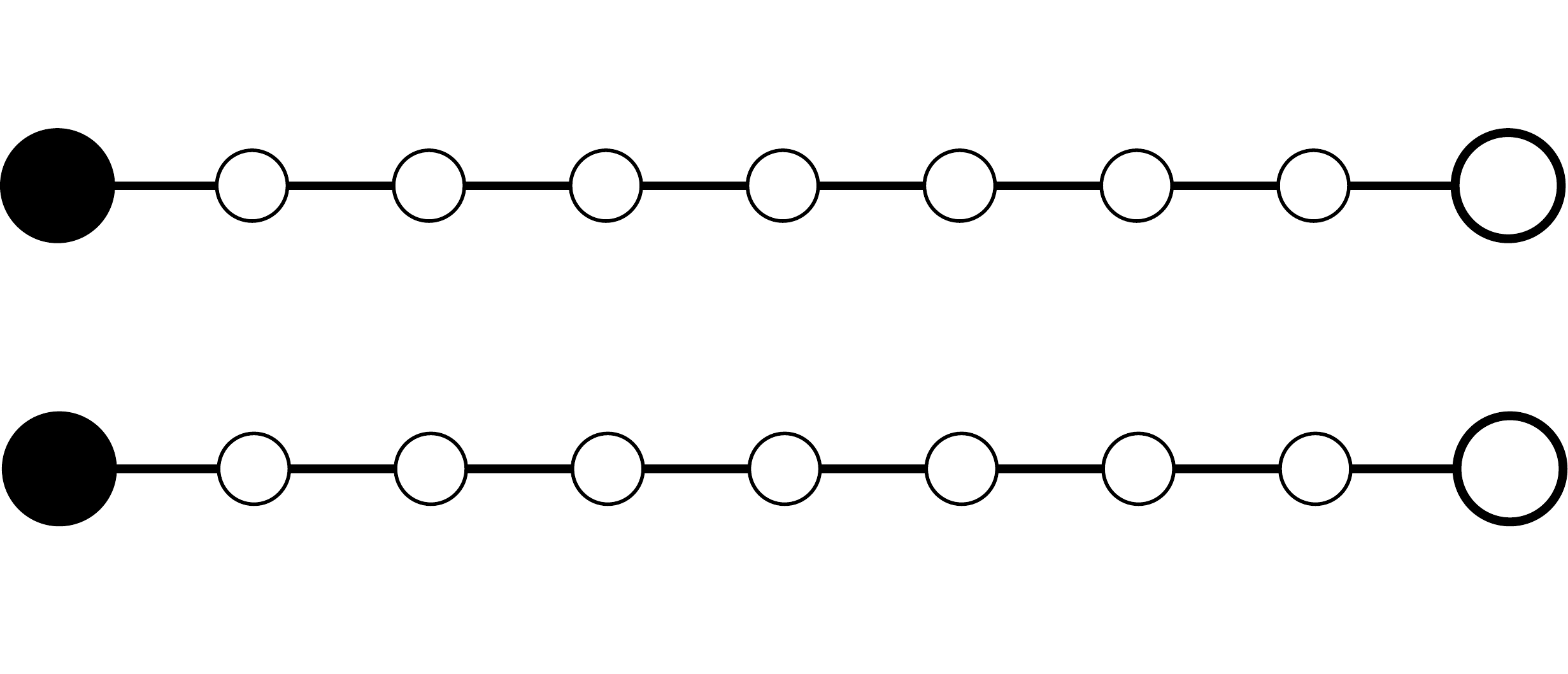

}\protect\caption{\label{fig:graphicalLM_Ex1_full}Graphical representation of $\mathcal{D}_{\varepsilon,n\, m}^{\Theta}\left(\bar{\theta}_{i}\right)$,
$K_{n,\alpha_{i}}\left(\bar{\theta}_{i}\right)$ and $J_{n,\alpha_{i}}\left(\bar{\theta}_{i}\right)$.}
\end{figure}
The proof follows the ideas used in demonstrating theorem \ref{thm:Dressed_FF_theorem_basic}.
The factor $\prod_{i}n_{i}!$ for the non-active singularities cancels
as before, but now one must sum over all possible positions of the
active singularity:
\begin{eqnarray}
\mathcal{D}_{\varepsilon}^{\Theta}\left(\bar{\theta}_{i}\right) & = & 2\pi\sum_{n,m=0}^{\infty}\sum_{\beta_{1},\dots,\beta_{n+m}}\int_{-\infty}^{\infty}\prod_{j=1}^{n+m}\frac{\mathrm{d}\theta_{j}}{2\pi\left[1+e^{\varepsilon_{\beta_{j}}\left(\theta_{j}\right)}\right]}\varphi_{12}\varphi_{23}\dots\varphi_{ni}\nonumber \\
 &  & \times\varphi_{i,n+1}\dots\varphi_{n+m-1,n+m}m_{\beta{}_{1}}m_{\beta_{n+m}}\cosh\left(\theta_{1,n+m}\right)=\sum_{n,m=0}^{\infty}\mathcal{D}_{\varepsilon,n\, m}^{\Theta}\left(\bar{\theta}_{i}\right)
\end{eqnarray}
where the sum runs for the number and the species nodes between the
active singularity and the end nodes (for $n=0$, the active singularity
is the end node on the left, while for $m=0$ it is the end node on
the right). The active singularity is marked by a black node in Figs.
\ref{fig:graphicalLM_Ex1} and \ref{fig:graphicalLM_Ex1_expand}.
$K_{n,\alpha_{i}}\left(\bar{\theta}_{i}\right)$ and $J_{n,\alpha_{i}}\left(\bar{\theta}_{i}\right)$
are represented in Fig. \ref{fig:graphicalLM_bEx1_fc_fs}; they are
equal to the contribution of the chain between one of the end nodes
and the active singularity. Multiplying them as in \ref{eq:theor_ff_Ex1}
it follows that
\begin{eqnarray}
\mathcal{D}_{\varepsilon,n\, m}^{\Theta}\left(\bar{\theta}_{i}\right) & = & K_{n,\alpha_{i}}\left(\bar{\theta}_{i}\right)K_{m,\alpha_{i}}\left(\bar{\theta}_{i}\right)-J_{n,\alpha_{i}}\left(\bar{\theta}_{i}\right)J_{m,\alpha_{i}}\left(\bar{\theta}_{i}\right)
\end{eqnarray}
and the summation in both of $f_{c/s}$ the result exactly reproduces
$\mathcal{D}_{\varepsilon}^{\Theta}\left(\bar{\theta}_{i}\right)$.
\emph{Q.e.d.}\end{proof}
\begin{thm}
\label{thm:Dressed_FF_theorem_Exmulti}The dressed form factor of
$\Theta$ with $N$ active singularities $\left\{ \bar{\theta}_{1},\dots,\bar{\theta}_{N}\right\} $
is

\begin{eqnarray}
\mathcal{D_{\varepsilon}^{O}}\left(\bar{\theta}_{1},\dots,\bar{\theta}_{N}\right) & = & \sum_{n_{1},\dots,n_{k}=0}^{\infty}\frac{1}{\prod_{i}n_{i}!}\int_{-\infty}^{\infty}\prod_{j=1}^{\tilde{N}}\frac{\mathrm{d}\theta_{j}}{2\pi\left[1+e^{\varepsilon_{\beta_{j}}\left(\theta_{j}\right)}\right]}\nonumber \\
 &  & \times F_{2n_{1},\dots,2n_{k},c}^{\mathcal{O}}\left(\bar{\theta}_{1},\dots,\bar{\theta}_{N},\theta_{1},\dots,\theta_{\tilde{N}}\right)\label{eq:FF_conn_Theta_theorem_Exmulti}
\end{eqnarray}
which is equal to
\begin{eqnarray}
\mathcal{D_{\varepsilon}^{O}}\left(\bar{\theta}_{1},\dots,\bar{\theta}_{N}\right) & = & \sum_{i\neq j}\sum_{\left\{ \alpha\right\} }\left[f_{c,\alpha_{i}}\left(\bar{\theta}_{i}\right)f_{c,\alpha_{j}}\left(\bar{\theta}_{j}\right)-f_{s,\alpha_{i}}\left(\bar{\theta}_{i}\right)f_{s,\alpha_{j}}\left(\bar{\theta}_{j}\right)\right]\nonumber \\
 &  & \times f_{\alpha_{1},\alpha_{i}}\left(\bar{\theta}_{i}\right)f_{\alpha_{2},\alpha_{1}}\left(\bar{\theta}_{1}\right)\dots f_{j,\alpha_{n}}\left(\bar{\theta}_{n}\right)\label{eq:FF_conn_Theta_Exmulti}
\end{eqnarray}
where $\left\{ \alpha\right\} =\left\{ 1,\dots,N\right\} \setminus\left\{ i,j\right\} $.\end{thm}
\begin{proof}
\begin{figure}
\centering{}\subfloat[\label{fig:graphicalLM_Exmulti} $\mathcal{D}_{\varepsilon}^{\Theta}\left(\bar{\theta}_{1},\dots,\bar{\theta}_{N}\right)$.]{\centering
\def\svgwidth{0.4\columnwidth}
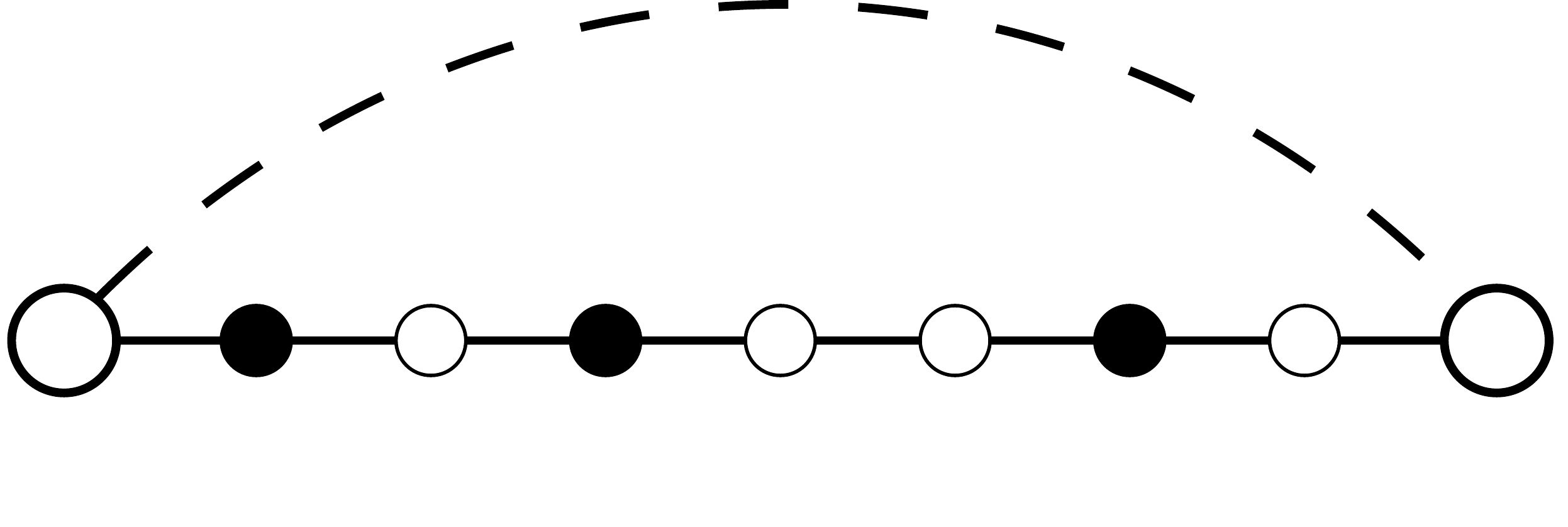}\hfill{}\subfloat[\label{fig:graphicalLM_Exmulti_expand}$\mathcal{D}_{\varepsilon}^{\Theta}\left(\bar{\theta}_{1},\dots,\bar{\theta}_{N}\right)$
expanded.]{\centering
\def\svgwidth{0.4\columnwidth}
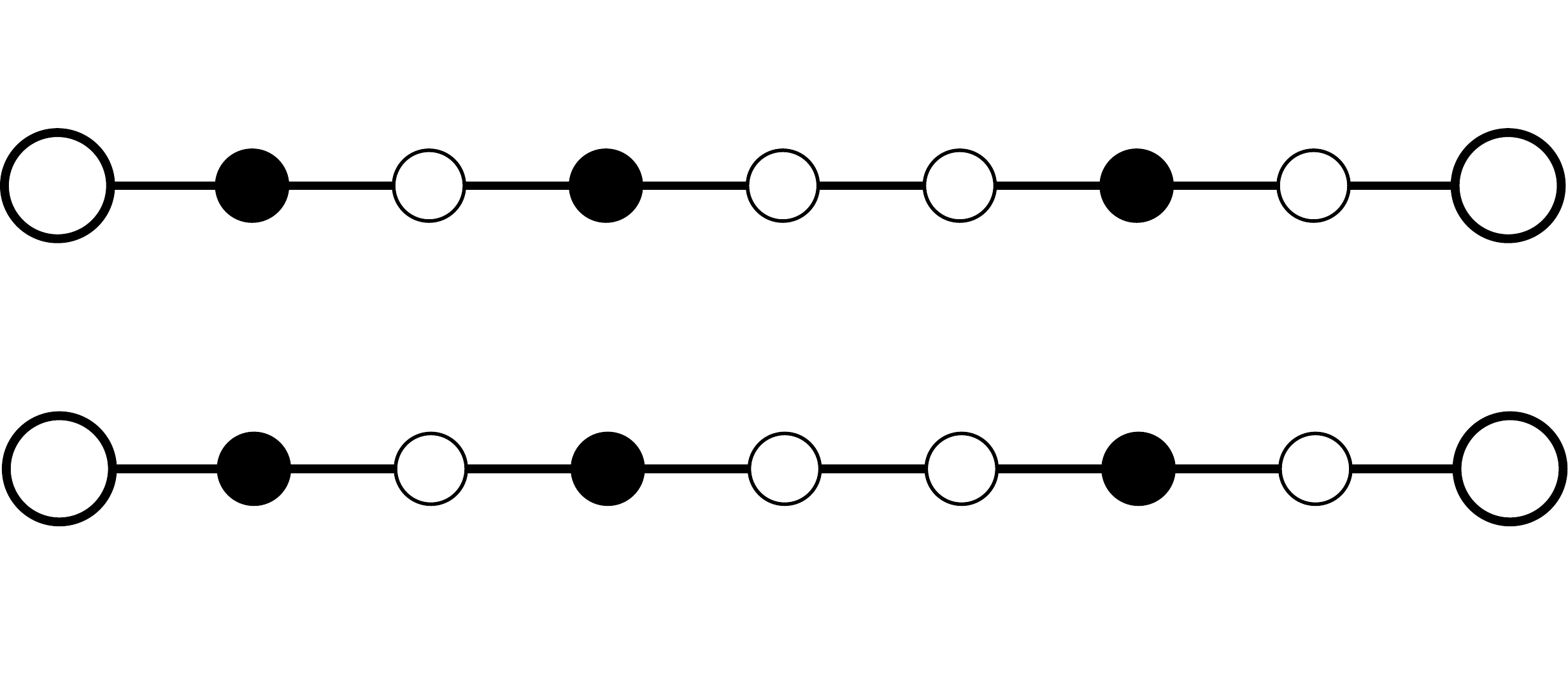

}\smallskip{}
\subfloat[\label{fig:graphicalLM_Exmulti_fij}$L_{n,ji}$ .]{\centering
\def\svgwidth{0.4\columnwidth}
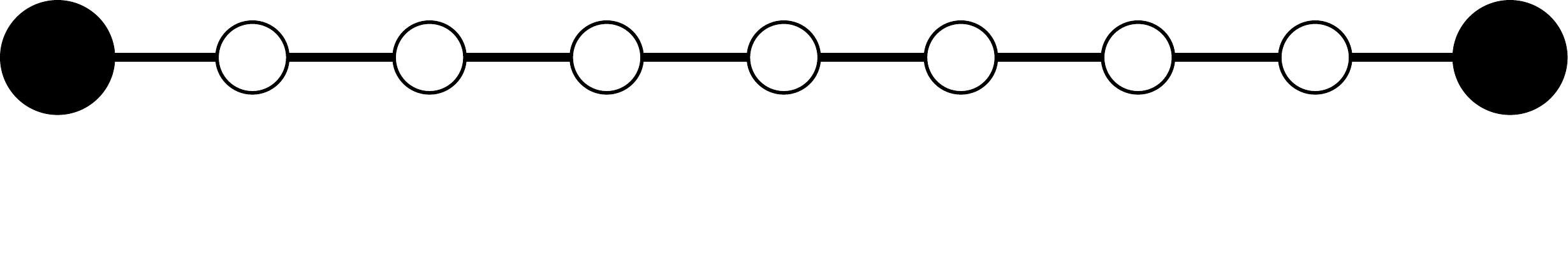

}\protect\caption{\label{fig:graphicalLM_Exmulti_full}Graphical representation of the
terms of $\mathcal{D}_{\varepsilon}^{\Theta}$$\left(\bar{\theta}_{1},\dots,\bar{\theta}_{N}\right)$
and $L_{n,ji}$ .}
\end{figure}
As in the proofs of the previous theorems, (\ref{eq:FF_conn_Theta_theorem_Exmulti})
can be organized into a sum over terms corresponding to individual
permutations of the active singularities. For a given permutation,
the contribution is the sum of graphs represented in Figs. \ref{fig:graphicalLM_Exmulti}
and \ref{fig:graphicalLM_Exmulti_expand}, where the number and type
of nodes separating end nodes and the active singularities is varying.

The functions $f_{\alpha_{q},\alpha_{p}}\left(\bar{\theta}_{p}\right)$
can be expanded as $f_{c/s}$ previously, using their definition in
(\ref{eq:fsfcfi})
\begin{eqnarray}
f_{\alpha_{q},\alpha_{p}}\left(\bar{\theta}_{p}\right) & = & \sum_{n=0}^{\infty}L_{n,pq}
\end{eqnarray}
where
\begin{eqnarray}
L_{0,pq} & = & \varphi_{pq}\nonumber \\
L_{n,pq} & = & \sum_{\beta_{1},\dots,\beta_{n}}\int\prod_{i}\frac{\mathrm{d}\theta_{i}}{2\pi\left[1+e^{\varepsilon_{\beta_{i}}\left(\theta_{i}\right)}\right]}\varphi_{\alpha_{p}\beta_{1}}\dots\varphi_{\beta_{n}\alpha_{q}}
\end{eqnarray}
$L_{n,pq}$ is represented in Fig. \ref{fig:graphicalLM_Exmulti_fij};
it generates all the possible contribution to Fig. \ref{fig:graphicalLM_Exmulti_expand}
between two active singularities. In a given permutation of the active
singularities let us take $\bar{\theta}_{i}$ and $\bar{\theta}_{j}$
as the two active singularities closest to the left/right end nodes;
then the terms 
\[
f_{c,\alpha_{i}}\left(\bar{\theta}_{i}\right)f_{c,\alpha_{j}}\left(\bar{\theta}_{j}\right)-f_{s,\alpha_{i}}\left(\bar{\theta}_{i}\right)f_{s,\alpha_{j}}\left(\bar{\theta}_{j}\right)
\]
generate all the contributions between the active singularities and
the ends, and 
\[
f_{\alpha_{1},\alpha_{i}}\left(\bar{\theta}_{i}\right)f_{\alpha_{2},\alpha_{1}}\left(\bar{\theta}_{1}\right)\dots f_{j,\alpha_{n}}\left(\bar{\theta}_{n}\right)
\]
generate all the contributions between the other active singularities.
Summing up for all the permutations of the active singularities proves
the theorem. \emph{Q.e.d.}
\end{proof}
Theorems \ref{thm:Dressed_FF_theorem_basic}, \ref{thm:Dressed_FF_theorem_Ex1}
and \ref{thm:Dressed_FF_theorem_Exmulti} prove the equivalence of
the form factor series (\ref{eq:LMExcited_general}) and the TBA equations
for $\left\langle \Theta\right\rangle $ in any excited state described
by the TBA system (\ref{eq:excited_state_TBA_and_energy}).

\section{Finite volume expectation values in the $T_{2}$ model \label{sec:Finite-volume-expectation-in-T2}}

For the numerical validation of the conjecture (\ref{eq:LMExcited_general})
we follow a similar strategy as we did for the Leclair-Mussardo conjecture
in \cite{Szecsenyi:2013gna}, where we chose the $T_{2}$ model for
the numerics. The $T_{2}$ model is the perturbation of the $\mathcal{M}_{2,7}$
conformal minimal model \cite{Belavin:1984vu} by the primary operator
$\Phi_{1,3}$. 

There are several advantages in this choice. First, all the data necessary
to calculate the series (\ref{eq:LMExcited_general}): the scattering
theory, the form factors (of primary fields) and the excited state
TBA equations, are known for this model. Second, the $T_{2}$ model
contains operators with known form factors that are not related to
the trace of the stress-energy tensor; the trace of the stress energy
tensor is not interesting since the conjecture is equivalent to the
TBA equation as shown in Section \ref{sec:Theta-TBA-LMEx}. Finally
for the $T_{2}$ model the Truncated Conformal Space Approach (TCSA)
improved by renormalization group methods \cite{Szecsenyi:2013gna}
gives an efficient way to evaluate the expectation values directly
solving the dynamics of the model. The TCSA was introduced in \cite{Yurov:1989yu}
while renormalization techniques were proposed in \cite{2007PhRvL..98n7205K,2008JSMTE..03..011F}
and further developed in \cite{2011arXiv1106.2448G,2012NuPhB.859..177W,2014JHEP...09..052L};
the development of related methods is now a very active field of investigation
\cite{2013NuPhB.877..457B,2014arXiv1409.1494C,2014arXiv1409.1581H,2014arXiv1412.3460R}.
For details on TCSA in the $T_{2}$ model and the renormalization
method we refer the interested reader to \cite{Szecsenyi:2013gna}.

\subsection{Excited state TBA equations for a single type-1 state}

\subsubsection{General form and solution of the excited TBA equations}

The simplest excited states in the excited TBA formalism for the $T_{2}$
model \cite{Dorey:1997rb} are those with a single type-1 particle.
For these excited states the TBA equations contain only two active
singularities of type-2 with $\eta_{1}=-1$, $\eta_{2}=1$: 
\begin{eqnarray}
\varepsilon_{a}\left(\theta\right) & = & m_{a}L\cosh\left(\theta\right)+\log\left(\frac{S_{a2}\left(\theta-\bar{\theta}_{1}\right)}{S_{a2}\left(\theta-\bar{\theta}_{2}\right)}\right)\label{eq:TBA_T2_particle1}\\
 &  & -\sum_{b}\int\frac{\mathrm{d}\theta'}{2\pi}\varphi_{ab}\left(\theta-\theta'\right)\log\left(1+e^{-\varepsilon_{b}\left(\theta'\right)}\right)\nonumber \\
e^{-\varepsilon_{2}\left(\bar{\theta}_{1/2}\right)} & = & -1\label{eq:TBA_T2_particle1_cond}\\
E_{\mathrm{TBA}}\left(L\right) & = & -im_{2}\left(\sinh\left(\bar{\theta}_{1}\right)-\sinh\left(\bar{\theta}_{2}\right)\right)\nonumber \\
 &  & -\sum_{a}\int\frac{\mathrm{d}\theta}{2\pi}m_{a}\cosh\left(\theta\right)\log\left(1+e^{-\varepsilon_{a}\left(\theta\right)}\right)\label{eq:TBA_T2_particle1_energ}
\end{eqnarray}
which are related by $\bar{\theta}_{2}=\bar{\theta}_{1}^{^{*}}$ for
states with nonzero momentum, where $*$ is the complex conjugation;
or $\bar{\theta}_{2}=-\bar{\theta}_{1}$ for zero momentum. 

In large volume (IR limit) the integral term is negligible in (\ref{eq:TBA_T2_particle1}),
and the $\cosh$ term goes to infinity with the volume, while the
value of $\varepsilon_{2}\left(\bar{\theta}_{1}\right)$ is finite,
which forces the imaginary part of the active singularity's position
to $\frac{i\pi}{10}$, since $S_{22}$ is singular around $\bar{\theta}_{1}-\bar{\theta}_{2}\sim\frac{i\pi}{5}$.
The position of the active singularity can be written as

\begin{eqnarray}
\bar{\theta}_{1} & = & \tilde{\theta}+i\left(\frac{\pi}{10}+\delta\right)
\end{eqnarray}
where $\tilde{\theta}$ and $\delta$ are real; $\delta$ is a correction
to the imaginary part that decays exponentially in the dimensionless
variable $m_{1}L$. Substituting this form into condition (\ref{eq:TBA_T2_particle1_cond})
and keeping only the first order corrections in $\delta$, the solution
for the position of the active singularity is 

\begin{eqnarray}
m_{1}L\sinh\tilde{\theta} & = & 2\pi s\nonumber \\
\delta & = & \cos\left(\pi s\right)\tan\left(\frac{3\pi}{10}\right)\tan^{2}\left(\frac{2\pi}{5}\right)e^{-m_{2}\cos\left(\frac{\pi}{10}\right)\sqrt{m_{1}^{2}L^{2}+\left(2\pi s\right)^{2}}}\label{eq:asymp_pole}
\end{eqnarray}
where $s$ is an integer number giving the momentum quantum number
of the state. Using this solution for TBA energy (\ref{eq:TBA_T2_particle1_energ})
and expanding the $\log\left(1+e^{-\varepsilon}\right)$ term in the
integral, the energy takes the following form\textcolor{red}{{} }
\begin{eqnarray}
E_{\mathrm{TBA}}\left(L\right) & = & E(L)-\mathcal{B}L\nonumber \\
 & = & 2\sin\left(\frac{\pi}{10}\right)m_{2}\cosh\left(\tilde{\theta}\right)+2m_{2}\cosh\left(\tilde{\theta}\right)\sin\left(\frac{\pi}{10}\right)\delta\nonumber \\
 &  & -\sum_{a}\int\frac{\mathrm{d}\theta}{2\pi}m_{a}\cosh\left(\theta\right)e^{-m_{a}L\cosh\left(\theta\right)}\frac{S_{a2}\left(\theta-\tilde{\theta}+i\frac{\pi}{10}\right)}{S_{a2}\left(\theta-\tilde{\theta}-i\frac{\pi}{10}\right)}\nonumber \\
 & = & \sqrt{m_{1}^{2}+\left(\frac{2\pi s}{L}\right)^{2}}+2\sqrt{m_{2}^{2}+\left(\frac{2\pi s}{L}\frac{m_{2}}{m_{1}}\right)^{2}}\cos\left(\frac{\pi}{10}\right)\nonumber \\
 &  & \times\cos\left(\pi s\right)\tan\left(\frac{3\pi}{10}\right)\tan^{2}\left(\frac{2\pi}{5}\right)e^{-m_{2}\cos\left(\frac{\pi}{10}\right)\sqrt{m_{1}^{2}L^{2}+\left(2\pi s\right)^{2}}}\nonumber \\
 &  & -\sum_{a}\int\frac{\mathrm{d}\theta}{2\pi}m_{a}\cosh\left(\theta\right)e^{-m_{a}L\cosh\left(\theta\right)}S_{a1}\left(\theta-\tilde{\theta}+i\frac{\pi}{2}\right)
\end{eqnarray}
where the bootstrap identity 
\begin{equation}
\frac{S_{a2}\left(\theta+i\frac{\pi}{10}\right)}{S_{a2}\left(\theta-i\frac{\pi}{10}\right)}=S_{a1}\left(\theta+i\frac{\pi}{2}\right)
\end{equation}
was used. The first term gives $L^{-1}$ corrections related to the
kinetic energy of the particle in finite volume, while the second
and third terms are the leading exponential corrections, the so-called
$\mu$ and $F$ terms \cite{1991NuPhB.362..329K}, which for a zero-momentum
particle were first derived by Lüscher in \cite{Luscher:1985dn}. 

In the small volume (UV) limit the energy is proportional to the effective
central charge of the state
\begin{eqnarray}
E_{\mathrm{TBA}}\left(L\right) & = & -\frac{6\pi}{L}c_{eff}\left(L\right)
\end{eqnarray}
which has the ultraviolet limit 
\begin{equation}
c_{eff}(0)=c-12(\Delta+\bar{\Delta})
\end{equation}
where $c=-68/7$ is the central charge of the minimal model $\mathcal{M}_{2,7}$
and $\Delta$, $\bar{\Delta}$ are the left/right conformal weights
of the state in the ultraviolet limit. Using the dilogarithm trick
introduced in \cite{Zamolodchikov:1989cf} one can confirm the expected
effective central charge for these states \cite{Dorey:1997rb}. In
Fig. \ref{fig:c_eff_TBA_TCSA_check} we plot $c_{eff}$ for the states
$s=0,1$ showing that the TBA results match perfectly with the TCSA
calculation and also reproduces the expected asymptotics. 

\textcolor{red}{}
\begin{figure}
\centering{}\textcolor{red}{\includegraphics[scale=0.53]{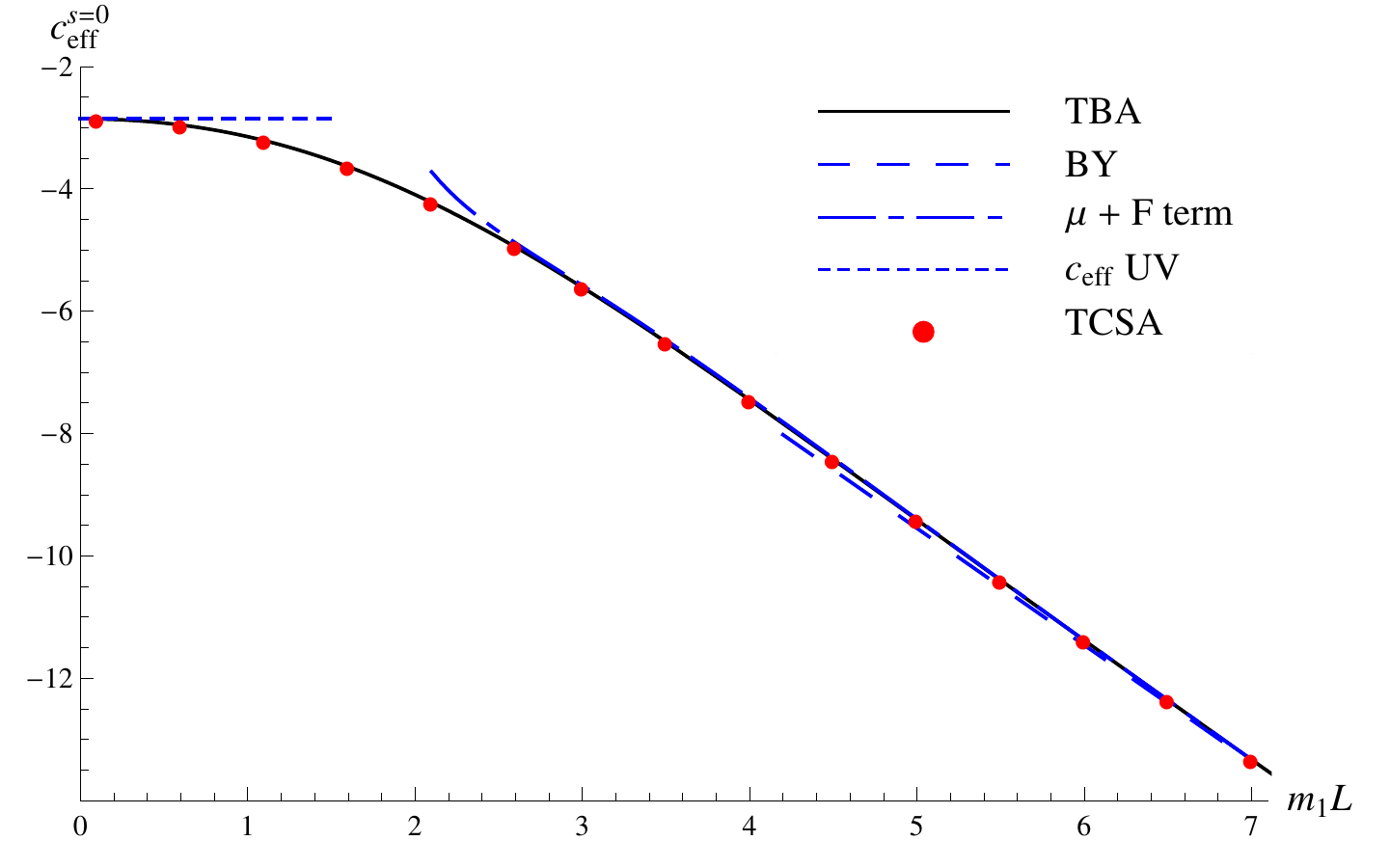}\includegraphics[scale=0.53]{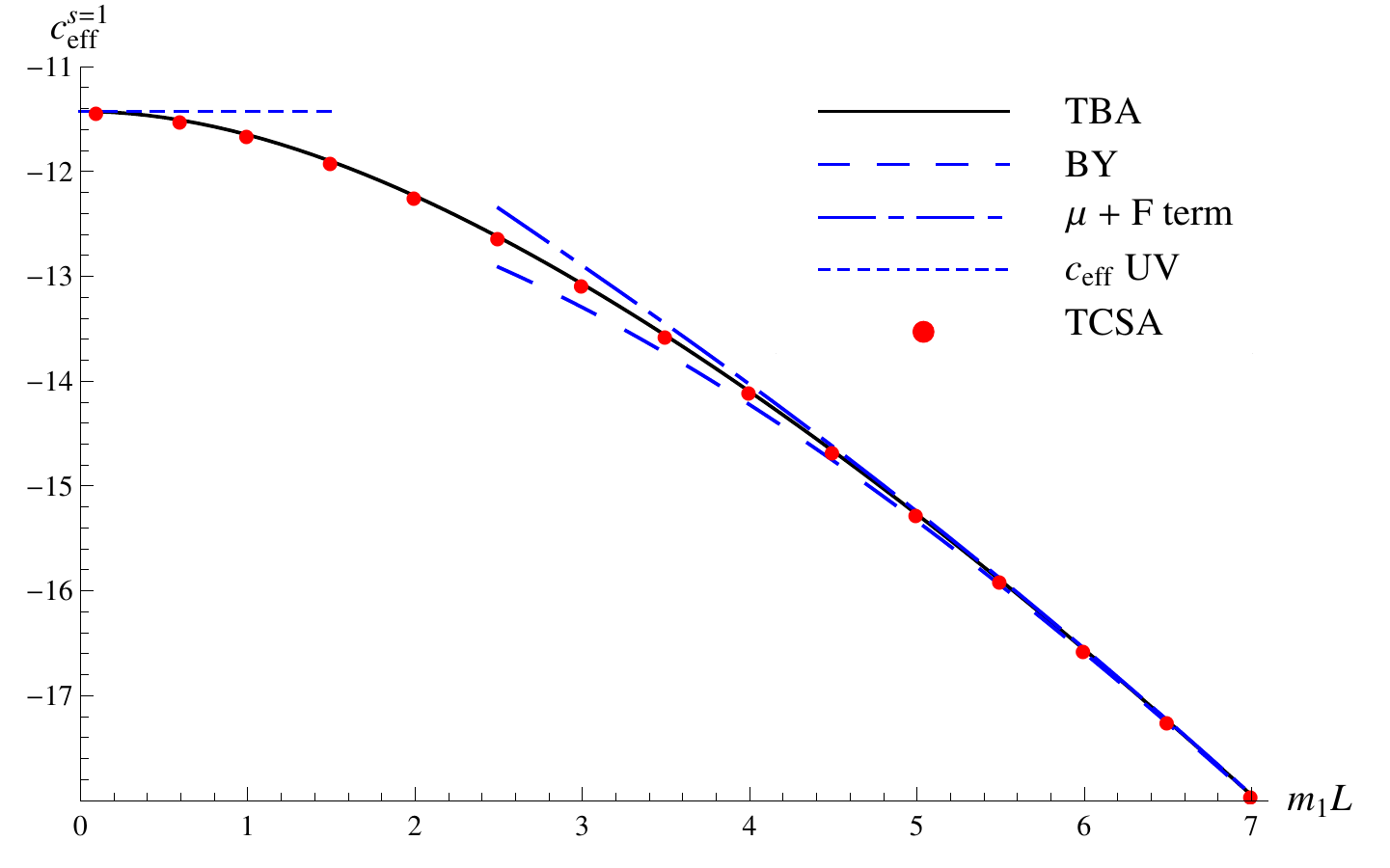}\protect\caption{\label{fig:c_eff_TBA_TCSA_check}Effective central charge for the
excited states with momentum quantum number $s=0,1$ from TBA and
TCSA, with the UV (CFT value) and IR asymptotics ($\mu$ and F term).}
}
\end{figure}

The excited state TBA equations are solved numerically by simultaneously
iterating eqns. (\ref{eq:TBA_T2_particle1}) and (\ref{eq:TBA_T2_particle1_cond})
in large volume, where the asymptotic of the pole position (\ref{eq:asymp_pole})
can be used as a starting point. Using this solution the volume is
decreased and the equations iterated, and continuing this process
the solution can be tracked to small volume. For $s\neq0$ the numerics
is straightforward to perform up to precision of order $10^{-12}$,
and all the ingredients to calculate the conjecture (\ref{eq:LMExcited_general})
can be readily constructed. For $s=0$ there exists a critical volume
$r_{c}$ under which it is necessary to be more careful with the numerical
calculation.

\subsubsection{Zero-momentum state: desingularization in small volume \label{sub:desingularisation}}

As described in details in \cite{Bazhanov:1996aq,Dorey:1996re,Dorey:1997rb},
for states containing a zero-momentum particle it can happen that
under a given critical volume $r_{c}$ some singularity ends up being
on the integration contour. To describe this situation, we recall
that the $Y$-system (\ref{eq:Ysystem}) gives relations between positions
where the functions $Y_{\alpha}=e^{\varepsilon_{\alpha}}$ take the
values 0 and $-1$, which are the logarithmic singular points of the
TBA equations. For the $T_{2}$ model the incidence matrix is given
by
\begin{eqnarray}
I^{\left[T_{2}\right]} & = & \begin{pmatrix}0 & 1\\
1 & 1
\end{pmatrix}\qquad\mbox{and}\qquad h=5
\end{eqnarray}
For the excited state containing a single type-1 particle with zero
momentum in large volume there are active singularities on the imaginary
axis at $\pm\bar{\theta}_{1}$ with 
\begin{equation}
\bar{\theta}_{1}=i\frac{\pi}{10}+i\delta
\end{equation}
with $\delta>0$ where $Y_{2}=-1$. From (\ref{eq:Ysystem}) it follows
that $Y_{2}=-1$ at $\bar{\theta}_{1}-i\frac{2\pi}{5}$ and $-\bar{\theta}_{1}+i\frac{2\pi}{5}$,
and $Y_{1}=Y_{2}=0$ at positions $\bar{\theta}_{1}-i\frac{\pi}{5}$
and $-\bar{\theta}_{1}+i\frac{\pi}{5}$. As the volume decreases,
the value of $\delta$ increases till at some critical value of the
volume given by $m_{1}L=r_{c}$ it reaches $\delta_{c}=\frac{\pi}{10}$.
At this point $\bar{\theta}_{1}=-\bar{\theta}_{1}+i\frac{2\pi}{5}$
resulting in a coincidence of singularities, and also of zeros. Decreasing
the volume to $m_{1}L<r_{c}$ results in the ``scattering`` of the
singularities on each other at right angles pushing them away from
the imaginary axis with fixed imaginary part in the form 
\begin{equation}
\bar{\theta}_{1}=i\frac{\pi}{5}+\alpha
\end{equation}
As a result, the zeros of $Y_{1}$ and $Y_{2}$ sit exactly on the
integration contour, making the equation for the pseudo-energy (\ref{eq:TBA_T2_particle1})
singular and leading to instabilities in the numerical solution of
the TBA equations.

One way to handle the numerical treatment of the problem is to shift
the integration contour, while an alternative is to rearrange the
self-consistent equations appropriately, which is called desingularization.
The latter approach relies on the relation \cite{Dorey:1997rb}
\begin{eqnarray*}
\varphi_{ad}\left(\theta\right) & = & -\varphi_{h}\left(\theta\right)I_{ad}^{\left[T_{2}\right]}+\sum_{b}\int\frac{\mathrm{d}\theta'}{2\pi}\varphi_{ab}\left(\theta-\theta'\right)\varphi_{h}\left(\theta'\right)I_{bd}^{\left[T_{2}\right]}
\end{eqnarray*}
\begin{eqnarray}
\varphi_{h}\left(\theta\right) & = & \frac{h}{2\cosh\left(\frac{h}{2}\theta\right)}=\pm i\partial_{\theta}\log\left(\sigma_{h}\left(\theta\pm i\frac{\pi}{h}\right)\right)\nonumber \\
\sigma_{h}\left(\theta\right) & = & \tanh\left(\frac{h}{4}\theta\right)
\end{eqnarray}
which allows the TBA equations to be recast as 
\begin{eqnarray}
\hat{\varepsilon}_{a}\left(\theta\right) & = & \varepsilon_{a}\left(\theta\right)-\log\left(\sigma\left(\theta'-\tilde{\theta}_{1}\right)\sigma\left(\theta'-\tilde{\theta}_{2}\right)\right)\nonumber \\
\hat{\varepsilon}_{a}\left(\theta\right) & = & m_{a}L\cosh\left(\theta\right)-\sum_{b}\int\frac{\mathrm{d}\theta'}{2\pi}\varphi_{ab}\left(\theta-\theta'\right)\log\left(\sigma\left(\theta'-\tilde{\theta}_{1}\right)\sigma\left(\theta'-\tilde{\theta}_{2}\right)+e^{-\hat{\varepsilon}_{b}\left(\theta'\right)}\right)\nonumber \\
e^{\hat{\varepsilon}_{a}\left(\bar{\theta}_{1}\right)} & = & i\coth\left(\frac{5}{2}\tilde{\theta}_{1}+i\frac{\pi}{4}\right)\nonumber \\
E_{\mathrm{TBA}}\left(L\right) & = & -\sum_{a}\int\frac{\mathrm{d}\theta}{2\pi}m_{a}\cosh\left(\theta\right)\log\left(\sigma\left(\theta-\tilde{\theta}_{1}\right)\sigma\left(\theta-\tilde{\theta}_{2}\right)+e^{-\hat{\varepsilon}_{b}\left(\theta\right)}\right)
\end{eqnarray}
where $\tilde{\theta}_{1}=\bar{\theta}_{1}-i\frac{\pi}{h}$, $\tilde{\theta}_{2}=\bar{\theta}_{2}+i\frac{\pi}{h}$.
These equations are regular and can be iterated in a stable way, however,
the available precision using double precision numbers drops to the
order of $10^{-10}$. Fortunately, that is still more than sufficient
for our purposes. 

To calculate the densities and $\left\langle \Theta\right\rangle $
under $r_{c}$ we need to desingularize $f_{s},f_{c},f_{i}$ in (\ref{eq:fsfcfi})
as well. It's easy to see, that the equation for $f_{s}$ and $f_{c}$
is regular under $r_{c}$ , since $\frac{1}{1+e^{\varepsilon_{\alpha}}}$
at the singularity position is regular, because $Y_{\alpha}=e^{\varepsilon_{\alpha}}=0$.
For $f_{i}$ the source term $\varphi$ is singular at $\tilde{\theta}_{1,2}$,
hence it is necessary to desingularize it:
\begin{eqnarray}
\hat{f}_{i,\alpha}\left(\theta\right) & = & f_{i,\alpha}\left(\theta\right)+\varphi_{h}\left(\theta-\bar{\theta}_{i}\right)\nonumber \\
\hat{f}_{i,\alpha}\left(\theta\right) & = & \sum_{\beta}\int\frac{\mathrm{d}\theta'}{2\pi}\varphi_{\alpha\beta}\left(\theta-\theta'\right)\frac{e^{\varepsilon_{\beta}\left(\theta'\right)}\varphi_{h}\left(\theta-\bar{\theta}_{i}\right)+\hat{f}_{i,\beta}\left(\theta'\right)}{1+e^{\varepsilon_{\beta}\left(\theta'\right)}}
\end{eqnarray}
For numerical calculations the form
\begin{eqnarray}
\hat{f}_{1,\alpha}\left(\theta\right) & = & \sum_{\beta}\int\frac{\mathrm{d}\theta'}{2\pi}\varphi_{\alpha\beta}\left(\theta-\theta'\right)\frac{\partial\sigma\left(\theta'-\tilde{\theta}_{1}\right)\sigma\left(\theta'-\tilde{\theta}_{2}\right)+e^{-\hat{\varepsilon}_{b}\left(\theta'\right)}\hat{f}_{1,\beta}\left(\theta'\right)}{\sigma\left(\theta'-\tilde{\theta}_{1}\right)\sigma\left(\theta'-\tilde{\theta}_{2}\right)+e^{-\hat{\varepsilon}_{b}\left(\theta'\right)}}\nonumber \\
\hat{f}_{2,\alpha}\left(\theta\right) & = & \sum_{\beta}\int\frac{\mathrm{d}\theta'}{2\pi}\varphi_{\alpha\beta}\left(\theta-\theta'\right)\frac{\sigma\left(\theta'-\tilde{\theta}_{1}\right)\partial\sigma\left(\theta'-\tilde{\theta}_{2}\right)+e^{-\hat{\varepsilon}_{b}\left(\theta'\right)}\hat{f}_{1,\beta}\left(\theta'\right)}{\sigma\left(\theta'-\tilde{\theta}_{1}\right)\sigma\left(\theta'-\tilde{\theta}_{2}\right)+e^{-\hat{\varepsilon}_{b}\left(\theta'\right)}}
\end{eqnarray}
is more convenient.

\subsection{Densities and the conjecture for states with a single type-1 particle }

The derivatives of the quantization condition can be written like
(\ref{eq:Density_general}) with the help of the definitions in (\ref{eq:fsfcfi}),
(\ref{eq:defN_I_N_phi}) for the single type-1 state 
\begin{eqnarray}
\frac{\partial\left(Q_{1},Q_{2}\right)}{\partial\left(\bar{\theta}_{1},\bar{\theta}_{2}\right)} & = & \left(\begin{array}{cc}
L\mathcal{N}_{1}+\mathcal{N}_{\varphi} & -\mathcal{N}_{\varphi}\\
-\mathcal{N}_{\varphi} & L\mathcal{N}_{2}+\mathcal{N}_{\varphi}
\end{array}\right)
\end{eqnarray}
where
\begin{eqnarray}
\mathcal{N}_{1} & = & -if_{s,2}\left(\bar{\theta}_{1}\right)=-im_{2}\sinh\left(\bar{\theta}_{1}\right)-i\sum_{\beta}\int\frac{\mathrm{d}\theta'}{2\pi}\varphi_{2\beta}\left(\bar{\theta}_{1}-\theta'\right)\frac{f_{s,\beta}\left(\theta'\right)}{1+e^{\varepsilon_{\beta}\left(\theta'\right)}}\nonumber \\
\mathcal{N}_{2} & = & if_{s,2}\left(\bar{\theta}_{2}\right)=im_{2}\sinh\left(\bar{\theta}_{2}\right)+i\sum_{\beta}\int\frac{\mathrm{d}\theta'}{2\pi}\varphi_{2\beta}\left(\bar{\theta}_{2}-\theta'\right)\frac{f_{s,\beta}\left(\theta'\right)}{1+e^{\varepsilon_{\beta}\left(\theta'\right)}}\nonumber \\
\mathcal{N}_{\varphi} & = & -f_{2,2}\left(\bar{\theta}_{1}\right)=-\varphi_{22}\left(\bar{\theta}_{1}-\bar{\theta}_{2}\right)-\sum_{\beta}\int\frac{\mathrm{d}\theta'}{2\pi}\varphi_{2\beta}\left(\bar{\theta}_{1}-\theta'\right)\frac{f_{2,\beta}\left(\theta'\right)}{1+e^{\varepsilon_{\beta}\left(\theta'\right)}}\nonumber \\
 & = & -f_{1,2}\left(\bar{\theta}_{2}\right)=-\varphi_{22}\left(\bar{\theta}_{2}-\bar{\theta}_{1}\right)-\sum_{\beta}\int\frac{\mathrm{d}\theta'}{2\pi}\varphi_{2\beta}\left(\bar{\theta}_{2}-\theta'\right)\frac{f_{1,\beta}\left(\theta'\right)}{1+e^{\varepsilon_{\beta}\left(\theta'\right)}}
\end{eqnarray}
For the case $s=0$ and $m_{1}L<r_{c}$ 
\begin{eqnarray}
\mathcal{N}_{\varphi} & = & -\hat{f}_{2,2}\left(\bar{\theta}_{1}\right)+\varphi_{h}\left(\bar{\theta}_{1}-\bar{\theta}_{2}\right)=+\varphi_{h}\left(\bar{\theta}_{1}-\bar{\theta}_{2}\right)\nonumber \\
 &  & -\sum_{\beta}\int\frac{\mathrm{d}\theta'}{2\pi}\varphi_{\alpha\beta}\left(\bar{\theta}_{1}-\theta'\right)\frac{\sigma\left(\theta'-\tilde{\theta}_{1}\right)\partial\sigma\left(\theta'-\tilde{\theta}_{2}\right)+e^{-\hat{\varepsilon}_{b}\left(\theta'\right)}\hat{f}_{1,\beta}\left(\theta'\right)}{\sigma\left(\theta'-\tilde{\theta}_{1}\right)\sigma\left(\theta'-\tilde{\theta}_{2}\right)+e^{-\hat{\varepsilon}_{b}\left(\theta'\right)}}\nonumber \\
 & = & -\hat{f}_{1,2}\left(\bar{\theta}_{2}\right)+\varphi_{h}\left(\bar{\theta}_{2}-\bar{\theta}_{1}\right)=+\varphi_{h}\left(\bar{\theta}_{2}-\bar{\theta}_{1}\right)\nonumber \\
 &  & -\sum_{\beta}\int\frac{\mathrm{d}\theta'}{2\pi}\varphi_{\alpha\beta}\left(\bar{\theta}_{2}-\theta'\right)\frac{\partial\sigma\left(\theta'-\tilde{\theta}_{1}\right)\sigma\left(\theta'-\tilde{\theta}_{2}\right)+e^{-\hat{\varepsilon}_{b}\left(\theta'\right)}\hat{f}_{1,\beta}\left(\theta'\right)}{\sigma\left(\theta'-\tilde{\theta}_{1}\right)\sigma\left(\theta'-\tilde{\theta}_{2}\right)+e^{-\hat{\varepsilon}_{b}\left(\theta'\right)}}
\end{eqnarray}
With the above densities the conjecture for the form factor series
(\ref{eq:LMExcited_general}) for this state takes the form 

\begin{eqnarray}
\left\langle \bar{\theta}_{1},\bar{\theta}_{2}\right|\mathcal{O}\left|\bar{\theta}_{1},\bar{\theta}_{2}\right\rangle _{L} & = & \sum_{n_{1},n_{2}=0}^{\infty}\frac{1}{n_{1}!n_{2}!}\int_{-\infty}^{\infty}\prod_{j=1}^{\tilde{N}}\frac{\mathrm{d}\theta_{j}}{2\pi\left[1+e^{\varepsilon_{\alpha_{j}}\left(\theta_{j}\right)}\right]}\Biggl[\nonumber \\
 &  & F_{2n_{1},2n_{2},c}^{\mathcal{O}}\left(\theta_{1},\dots,\theta_{\tilde{N}}\right)+\frac{1}{L^{2}\mathcal{N}_{1}\mathcal{N}_{2}+\mathcal{N}_{\varphi}L\left(\mathcal{N}_{1}+\mathcal{N}_{2}\right)}\Biggl\{\nonumber \\
 &  & +\left(L\mathcal{N}_{1}+\mathcal{N}_{\varphi}\right)F_{2n_{1},2n_{2}+2,c}^{\mathcal{O}}\left(\bar{\theta}_{1},\theta_{1},\dots,\theta_{\tilde{N}}\right)\nonumber \\
 &  & +\left(L\mathcal{N}_{2}+\mathcal{N}_{\varphi}\right)F_{2n_{1},2n_{2}+2,c}^{\mathcal{O}}\left(\bar{\theta}_{2},\theta_{1},\dots,\theta_{\tilde{N}}\right)\nonumber \\
 &  & +F_{2n_{1},2n_{2}+4,c}^{\mathcal{O}}\left(\bar{\theta}_{1},\bar{\theta}_{2},\theta_{1},\dots,\theta_{\tilde{N}}\right)\Biggr\}\Biggr]\label{eq:LMEx_T2}
\end{eqnarray}

\section{Numerical results \label{sec:Numerical-results}}

The last ingredient needed to compute the form factor series (\ref{eq:LMExcited_general})
is the numerical evaluation of the connected diagonal form factors
of the $T_{2}$ model. This is rather nontrivial to perform in a sufficiently
quick and numerically stable way for large number of particles. Since
this is a task that is also important for evaluating the spectral
series for correlation functions at finite temperature \cite{Pozsgay2008b,Pozsgay2010a,Szecsenyi2012,Buccheri:2013gta},
we describe the required tricks in Appendix \ref{sec:form_factors}.
The procedure can be straightforwardly generalized to connected diagonal
form factors in other integrable models. 

In large volume a rough estimate for the magnitude of the terms in
the series comes from the behaviour of the filling fractions 
\[
\prod_{i}\frac{1}{1+e^{\varepsilon_{\alpha_{j}}\left(\theta_{j}\right)}}\lesssim e^{-\left(\sum_{i}m_{i}\right)L}
\]
where $m_{i}$ are the masses of the particles contained in the given
state. Using this estimate we can identify the terms of the series
that give the dominant contribution to the expectation value. However,
with decreasing volume the ordering of the terms can change depending
on the behaviour of the pseudo-energy functions; in addition, to maintain
accuracy it is necessary to add progressively more terms. As a result,
the form factor series (\ref{eq:LMExcited_general}) is effectively
an infrared (low energy/large volume) expansion for the expectation
value. 

For our numerical calculations we implemented the terms with less
than 4 integrals, since for higher terms the number of integrals and
the size of the form factors makes the numerical integration too time-consuming;
in addition, the terms incorporated already show an excellent agreement
with the conjecture. Table \ref{tab:states_in_calculation} shows
the terms calculated for numerics.

\begin{table}
\begin{centering}
$\begin{array}[t]{|c||c|c|c|c|c|c|c|c|}
\hline \mbox{label} & 0 & 1 & 2 & 11 & 12 & 111 & 22 & 112\\
\hline\hline \mbox{\#\,\ of\,\ type-1\,\ integrals} & 0 & 1 & 0 & 2 & 1 & 3 & 0 & 2\\
\hline \mbox{\#\,\ of\,\ type-2\,\ integrals} & 0 & 0 & 1 & 0 & 1 & 0 & 2 & 1\\
\hline \mbox{magnitude/L} & 0 & m_{1} & m_{2} & 2m_{1} & m_{1}+m_{2} & 3m_{1} & 2m_{2} & 2m_{1}+m_{2}
\\\hline \end{array}$ \protect\caption{\label{tab:states_in_calculation}The terms incorporated from the
form factor series into the numerical calculation.}

\par\end{centering}

\end{table}

In the $T_{2}$ model there are two primary operators, namely $\Phi_{1,3}$
and $\Phi_{1,2}$. $\Phi_{1,3}$ is the operator perturbing the UV
limit of the model, hence it's proportional to $\Theta$ 
\begin{eqnarray}
\Theta & = & 2\pi\lambda\left(2h_{1,3}-2\right)\Phi_{1,3}
\end{eqnarray}
where $h_{1,3}$ is the conformal weight of the $\Phi_{1,3}$ and
$\lambda$ is the coupling constant of the perturbation that is proportional
to the mass gap \cite{Fateev:1993av}
\begin{eqnarray}
\lambda & = & \kappa m_{1}^{2-2h_{1,3}}
\end{eqnarray}
with 
\begin{eqnarray}
\kappa & = & \text{\textminus}0.04053795542...
\end{eqnarray}
Because of this, the form factor series for $\Phi_{1,3}$ is equivalent
to the TBA equation as proved in Section \ref{sec:Theta-TBA-LMEx},
and the numerical calculation for $\Phi_{1,3}$ is therefore not a
real further test for the general validity of the form factor series
(\ref{eq:LMExcited_general}). However it is still useful since with
its expectation value known from TBA equations one can get an independent
check of the numerical precision of TCSA, and the convergence of the
form factor series. For $\Phi_{1,2}$ there is only TCSA and the form
factor series, with the numerical deviation for $\Phi_{1,3}$ setting
the expected precision for the agreement between them. 

For the numerical integration we used the Cuba library routines \cite{2005CoPhC.168...78H},
called from inside Wolfram Mathematica \cite{Mathematica}.

\subsection{Moving one-particle state, $s=1$ }

For the moving type-1 excited state with momentum quantum number $s=1$,
Figure \ref{fig:Phi13_spin1_vev} shows the expectation value $\!_{1}\left\langle \left\{ 1\right\} \right|\Phi_{1,3}\left|\left\{ 1\right\} \right\rangle _{1}$
calculated with RG-extrapolated TCSA, from TBA together with the results
from the form factor series (\ref{eq:LMEx_T2}) obtained by adding
progressively more terms. The precision of the TBA is of the order
$10^{-12}$ and comparing it with the TCSA data, we find that the
precision of the RG-extrapolated TCSA is of order $10^{-6}-10^{-7}$.
Table \ref{tab:Phi13_spin1_diffTCSA} shows the difference between
the form factor series with different terms involved and the TCSA
data. For volume $m_{1}L>5$ the difference between the form factor
series up to and including the $112$ term, and the TCSA is in the
order of the TCSA error, and including more terms make the agreement
better for smaller volume as well. 

\begin{figure}
\centering{}\includegraphics{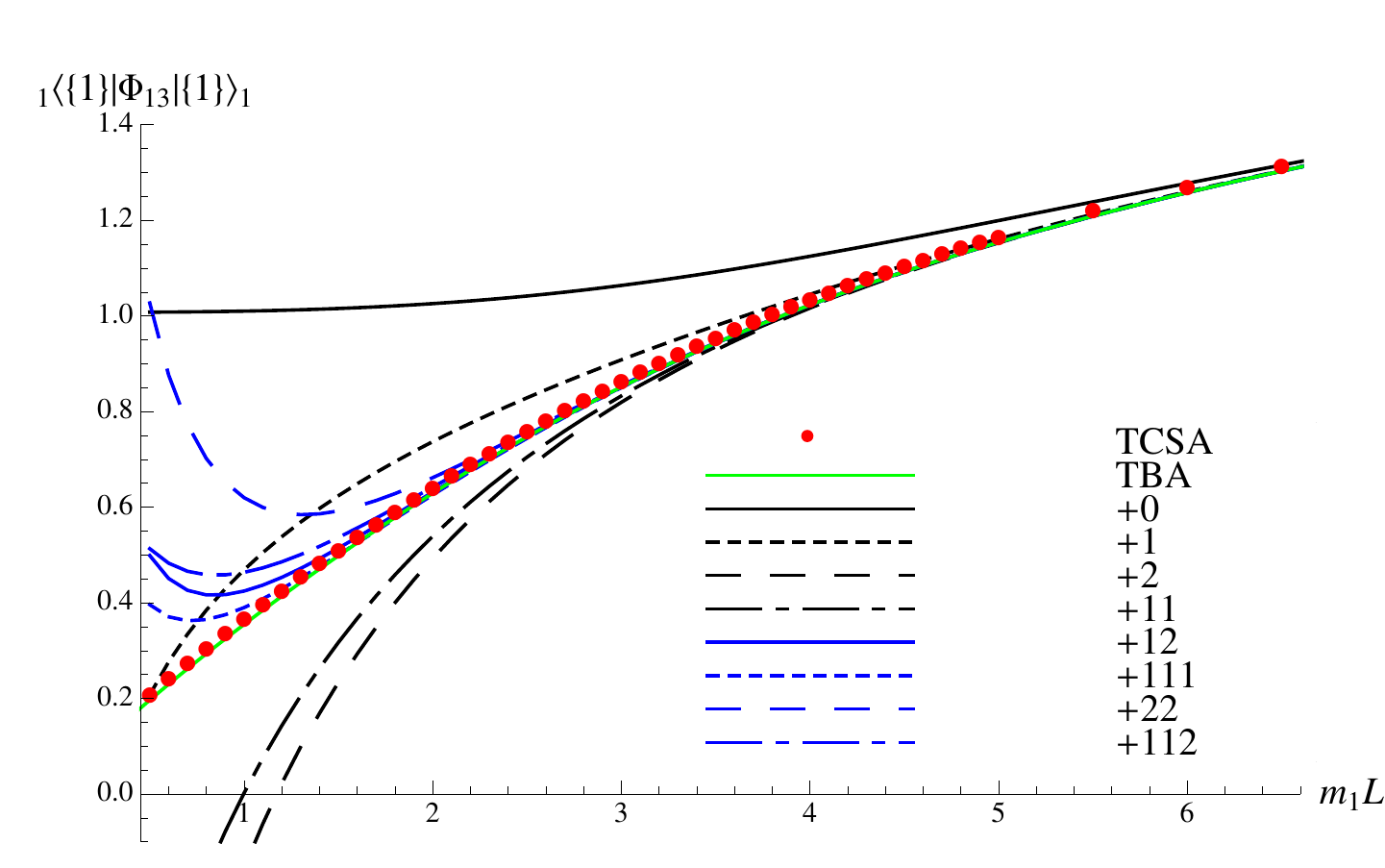}\protect\caption{\label{fig:Phi13_spin1_vev} $\!_{1}\!\!\left\langle \left\{ 1\right\} \right|\Phi_{1,3}\left|\left\{ 1\right\} \right\rangle _{1}$
evaluated by the form factor series including different contributions,
RG extrapolated TCSA and TBA. Here and in all subsequent plots matrix
elements are given in units of $m_{1}$.}
\end{figure}

For $\Phi_{1,2}$ the results for the quantity $i\!_{1}\!\!\left\langle \left\{ 1\right\} \right|\Phi_{1,2}\left|\left\{ 1\right\} \right\rangle _{1}$
are shown in Figure \ref{fig:Phi12_spin1_vev}, and the difference
between the form factor series and the TCSA is given in Table \ref{tab:Phi12_spin1_diffTCSA}.
We note that since from (\ref{eq:12_13_exactvevs}) it follows that
the matrix elements of $\Phi_{1,2}$ are imaginary, here and in all
subsequent figures and tables concerning $\Phi_{1,2}$ we multiply
all data by $i$. The form factor series shows excellent agreement
with the TCSA for volume $m_{1}L>5$, and again including more terms
the agreement is better for smaller volumes. As noted before, the
correctness of the form factor series does not follow from TBA, hence
this is a nontrivial verification of the form factor series.

\begin{figure}
\centering{}\includegraphics{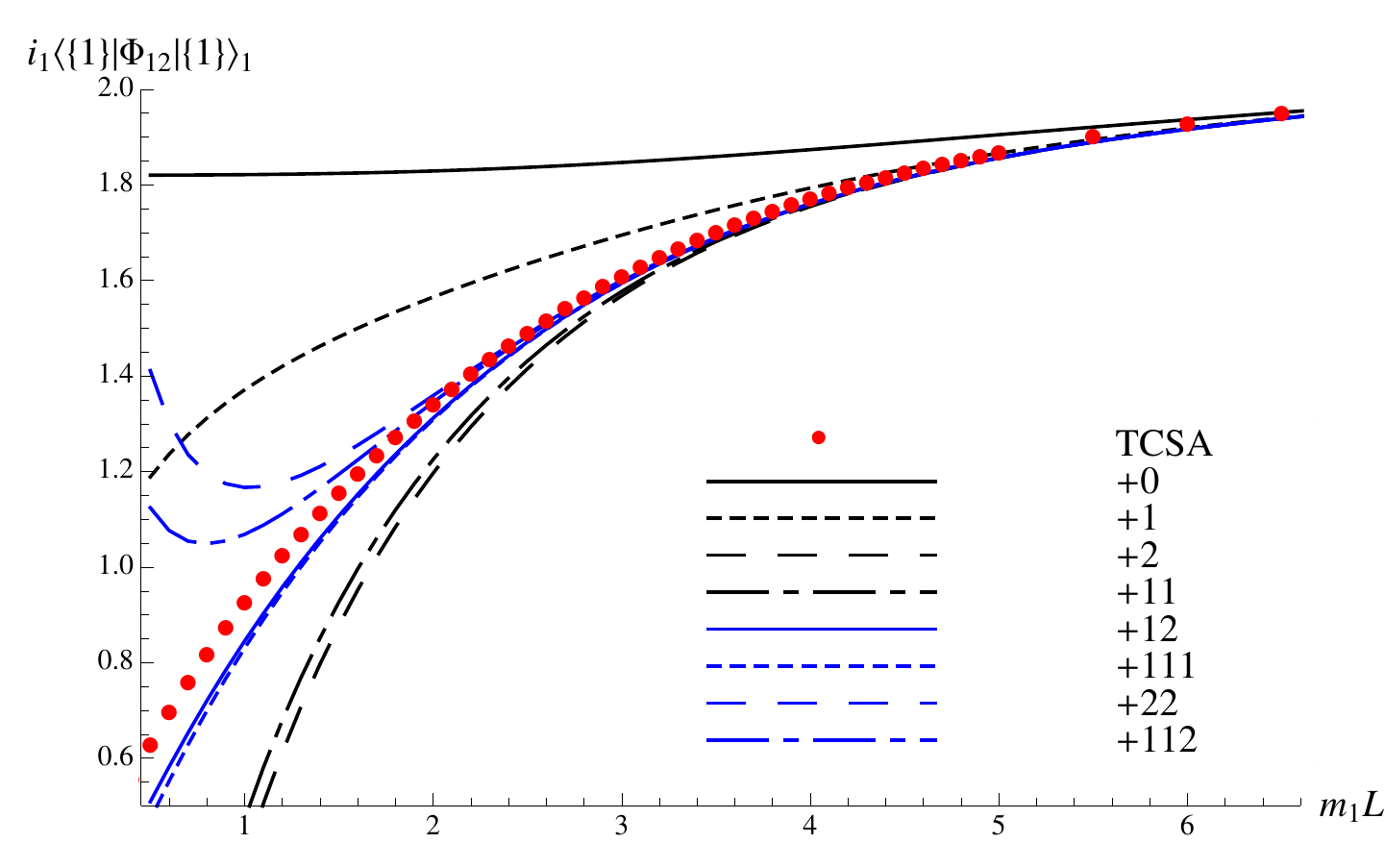}\protect\caption{\label{fig:Phi12_spin1_vev} $i\,_{1}\!\!\left\langle \left\{ 1\right\} \right|\Phi_{1,2}\left|\left\{ 1\right\} \right\rangle _{1}$
evaluated by the form factor series including different contributions
and RG-extrapolated TCSA. }
\end{figure}
\begin{table}
\begin{centering}
$\begin{array}{|c||c|c|c|c|c|c|c|c|}
\hline m_{1}L & 0 & +1 & +2 & +11 & +12 & +111 & +22 & +112\\
\hline\hline 0.5 & -0.8115 & -1\cdot10^{-2} & 1.039 & 0.7047 & -0.3024 & -0.1995 & -0.8307 & -0.3166\\
\hline 1 & -0.6554 & -0.1143 & 0.5094 & 0.3528 & -7\cdot10^{-2} & -4\cdot10^{-2} & -0.2652 & -0.1089\\
\hline 1.5 & -0.5173 & -0.1256 & 0.2649 & 0.1813 & -2\cdot10^{-2} & -2\cdot10^{-3} & -9\cdot10^{-2} & -4\cdot10^{-2}\\
\hline 2 & -0.3966 & -0.1079 & 0.1362 & 9\cdot10^{-2} & -2\cdot10^{-3} & 4\cdot10^{-3} & -3\cdot10^{-2} & -1\cdot10^{-2}\\
\hline 2.5 & -0.2947 & -8\cdot10^{-2} & 7\cdot10^{-2} & 4\cdot10^{-2} & 5\cdot10^{-4} & 3\cdot10^{-3} & -1\cdot10^{-2} & -4\cdot10^{-3}\\
\hline 3 & -0.2124 & -6\cdot10^{-2} & 3\cdot10^{-2} & 2\cdot10^{-2} & 6\cdot10^{-4} & 2\cdot10^{-3} & -3\cdot10^{-3} & -1\cdot10^{-3}\\
\hline 3.5 & -0.1489 & -4\cdot10^{-2} & 1\cdot10^{-2} & 8\cdot10^{-3} & 3\cdot10^{-4} & 6\cdot10^{-4} & -9\cdot10^{-4} & -3\cdot10^{-4}\\
\hline 4 & -0.102 & -2\cdot10^{-2} & 6\cdot10^{-3} & 3\cdot10^{-3} & 1\cdot10^{-4} & 2\cdot10^{-4} & -2\cdot10^{-4} & -5\cdot10^{-5}\\
\hline 4.5 & -7\cdot10^{-2} & -1\cdot10^{-2} & 3\cdot10^{-3} & 1\cdot10^{-3} & 4\cdot10^{-5} & 8\cdot10^{-5} & -5\cdot10^{-5} & -1\cdot10^{-5}\\
\hline 5 & -5\cdot10^{-2} & -7\cdot10^{-3} & 1\cdot10^{-3} & 4\cdot10^{-4} & 1\cdot10^{-5} & 2\cdot10^{-5} & -1\cdot10^{-5} & -2\cdot10^{-6}\\
\hline 5.5 & -3\cdot10^{-2} & -4\cdot10^{-3} & 5\cdot10^{-4} & 1\cdot10^{-4} & 4\cdot10^{-6} & 7\cdot10^{-6} & -3\cdot10^{-6} & -4\cdot10^{-7}\\
\hline 6 & -2\cdot10^{-2} & -2\cdot10^{-3} & 2\cdot10^{-4} & 4\cdot10^{-5} & 8\cdot10^{-7} & 2\cdot10^{-6} & -7\cdot10^{-7} & -2\cdot10^{-7}\\
\hline 6.5 & -1\cdot10^{-2} & -1\cdot10^{-3} & 7\cdot10^{-5} & 1\cdot10^{-5} & 4\cdot10^{-7} & 6\cdot10^{-7} & 3\cdot10^{-8} & 1\cdot10^{-7}\\
\hline 7 & -8\cdot10^{-3} & -5\cdot10^{-4} & 3\cdot10^{-5} & 5\cdot10^{-6} & 4\cdot10^{-7} & 4\cdot10^{-7} & 3\cdot10^{-7} & 3\cdot10^{-7}
\\\hline \end{array}$
\par\end{centering}

\protect\caption{\label{tab:Phi13_spin1_diffTCSA}The difference between the evaluations
of $\!_{1}\!\!\left\langle \left\{ 1\right\} \right|\Phi_{1,3}\left|\left\{ 1\right\} \right\rangle _{1}$
from the RG-extrapolated TCSA and the form factor series, depending
on the multi-particle contributions included in the latter.}
\end{table}
\begin{table}
\begin{centering}
$\begin{array}{|c||c|c|c|c|c|c|c|c|}
\hline m_{1}L & 0 & +1 & +2 & +11 & +12 & +111 & +22 & +112\\
\hline\hline 0.5 & -1.203 & -0.5712 & 1.113 & 0.9523 & 0.1086 & 0.1457 & -0.7935 & -0.5066\\
\hline 1 & -0.9067 & -0.4563 & 0.5275 & 0.4431 & 7\cdot10^{-2} & 8\cdot10^{-2} & -0.252 & -0.1539\\
\hline 1.5 & -0.6803 & -0.3379 & 0.2661 & 0.2171 & 4\cdot10^{-2} & 4\cdot10^{-2} & -9\cdot10^{-2} & -5\cdot10^{-2}\\
\hline 2 & -0.5 & -0.2361 & 0.1325 & 0.1038 & 2\cdot10^{-2} & 2\cdot10^{-2} & -3\cdot10^{-2} & -2\cdot10^{-2}\\
\hline 2.5 & -0.3581 & -0.1561 & 6\cdot10^{-2} & 5\cdot10^{-2} & 7\cdot10^{-3} & 9\cdot10^{-3} & -9\cdot10^{-3} & -5\cdot10^{-3}\\
\hline 3 & -0.2498 & -1\cdot10^{-1} & 3\cdot10^{-2} & 2\cdot10^{-2} & 3\cdot10^{-3} & 3\cdot10^{-3} & -3\cdot10^{-3} & -1\cdot10^{-3}\\
\hline 3.5 & -0.1701 & -6\cdot10^{-2} & 1\cdot10^{-2} & 8\cdot10^{-3} & 1\cdot10^{-3} & 1\cdot10^{-3} & -8\cdot10^{-4} & -3\cdot10^{-4}\\
\hline 4 & -0.1136 & -3\cdot10^{-2} & 5\cdot10^{-3} & 3\cdot10^{-3} & 3\cdot10^{-4} & 4\cdot10^{-4} & -2\cdot10^{-4} & -6\cdot10^{-5}\\
\hline 4.5 & -7\cdot10^{-2} & -2\cdot10^{-2} & 2\cdot10^{-3} & 1\cdot10^{-3} & 9\cdot10^{-5} & 1\cdot10^{-4} & -5\cdot10^{-5} & -1\cdot10^{-5}\\
\hline 5 & -5\cdot10^{-2} & -1\cdot10^{-2} & 9\cdot10^{-4} & 4\cdot10^{-4} & 3\cdot10^{-5} & 3\cdot10^{-5} & -1\cdot10^{-5} & -3\cdot10^{-6}\\
\hline 5.5 & -3\cdot10^{-2} & -5\cdot10^{-3} & 4\cdot10^{-4} & 1\cdot10^{-4} & 6\cdot10^{-6} & 8\cdot10^{-6} & -3\cdot10^{-6} & -2\cdot10^{-6}\\
\hline 6 & -2\cdot10^{-2} & -3\cdot10^{-3} & 1\cdot10^{-4} & 4\cdot10^{-5} & 2\cdot10^{-7} & 8\cdot10^{-7} & -2\cdot10^{-6} & -2\cdot10^{-6}\\
\hline 6.5 & -1\cdot10^{-2} & -1\cdot10^{-3} & 5\cdot10^{-5} & 1\cdot10^{-5} & -1\cdot10^{-6} & -1\cdot10^{-6} & -2\cdot10^{-6} & -2\cdot10^{-6}\\
\hline 7 & -8\cdot10^{-3} & -6\cdot10^{-4} & 2\cdot10^{-5} & 2\cdot10^{-6} & -2\cdot10^{-6} & -2\cdot10^{-6} & -2\cdot10^{-6} & -2\cdot10^{-6}
\\\hline \end{array}$
\par\end{centering}

\centering{}\protect\caption{\label{tab:Phi12_spin1_diffTCSA}The difference between the evaluations
of $i\,_{1}\!\!\left\langle \left\{ 1\right\} \right|\Phi_{1,2}\left|\left\{ 1\right\} \right\rangle _{1}$
from the RG-extrapolated TCSA and the form factor series, depending
on the multi-particle contributions included in the latter.}
\end{table}

\subsection{Zero-momentum one-particle state, $s=0$ }

As seen for the $s=1$ case the form factor series reproduce the expectation
value of local operators with very good precision in large volume,
even by including only few terms from the series. The expectation
is that for any state in small volume it is necessary to include higher
contributions of the series, but for any desired accuracy a finite
number of them is sufficient. 

As shown below, this expectation is challenged by the nontrivial transition
in the TBA equation for standing state at $r_{c}$. Figure \ref{fig:Phi13_spin0_vev}
and \ref{fig:Phi12_spin0_vev} shows the result for the expectation
values of $\Phi_{1,3}$ and $\Phi_{1,2}$, while Table \ref{tab:Phi13_spin0_diffTCSA}
and \ref{tab:Phi12_spin0_diffTCSA} list the numerical deviations
from RG-extrapolated TCSA. 

For large volume ($m_{1}L\gtrsim6$) the agreement between the form
factor series and the TCSA is again excellent. However, towards the
critical volume $\left(r_{c}\sim2.66\right)$ the terms of the form
factor series tend to diverge. This can be understood from the fact
that the total density of the states 
\begin{equation}
\rho_{tot}=\det\mathcal{K}=L^{2}\mathcal{N}_{1}\mathcal{N}_{2}+\mathcal{N}_{\varphi}L\left(\mathcal{N}_{1}+\mathcal{N}_{2}\right)
\end{equation}
which the denominator of the form factor series (\ref{eq:LMEx_T2}),
is zero at $r_{c}$. Figure \ref{fig:rho_tot_spin0} shows the behaviour
of $\rho_{tot}$ around $r_{c}$; fitting the location where $\rho_{tot}\left(\tilde{r}_{c}\right)=0$,
the result is 
\begin{equation}
\tilde{r}_{c}=2.6646510318
\end{equation}
which is perfect agreement with the value of $r_{c}$ obtained in
\cite{Dorey:1997rb}. 

The reason for $\rho_{tot}$ vanishing at $r_{c}$ can be understood
from the excited TBA. Since the active singularities coincide at this
point, the density which is the Jacobi determinant of the quantization
condition for the active singularities, is zero due to the degeneracy.
Such singularities of the density were observed previously for the
finite volume form factors formula in \cite{Szecsenyi:2013gna}; however
in that case including the exponential corrections resolved (or at
least shifted) the singularity. The present situation is different
as all exponential corrections to the density are already included.
To resolve the singularity it would be necessary to include every
term of the form factor series, to compensate the zero of the denominator. 

This conclusion is also supported by the behaviour of the pseudo-energies
and the filling fractions around $r_{c}$. Approaching $r_{c}$ the
filling fractions no more suppress the higher order terms in the series
and the ordering of terms by their magnitude is not valid anymore,
i.e. every terms is important in the series. This is consistent with
the procedure of desingularization, whereby to describe the excited
state level with the TBA equation under $r_{c}$ it was necessary
to redefine the pseudo-energy to have a form which is finite and convergent
under iterations. 

For the form factor series a similar rearrangement is necessary close
to and under $r_{c}$. Unfortunately such a rearrangement is not yet
known, and this sets the practical validity of the form factor series
to IR regions where no nontrivial transitions occur in the TBA equations. 

\begin{figure}
\centering{}\includegraphics{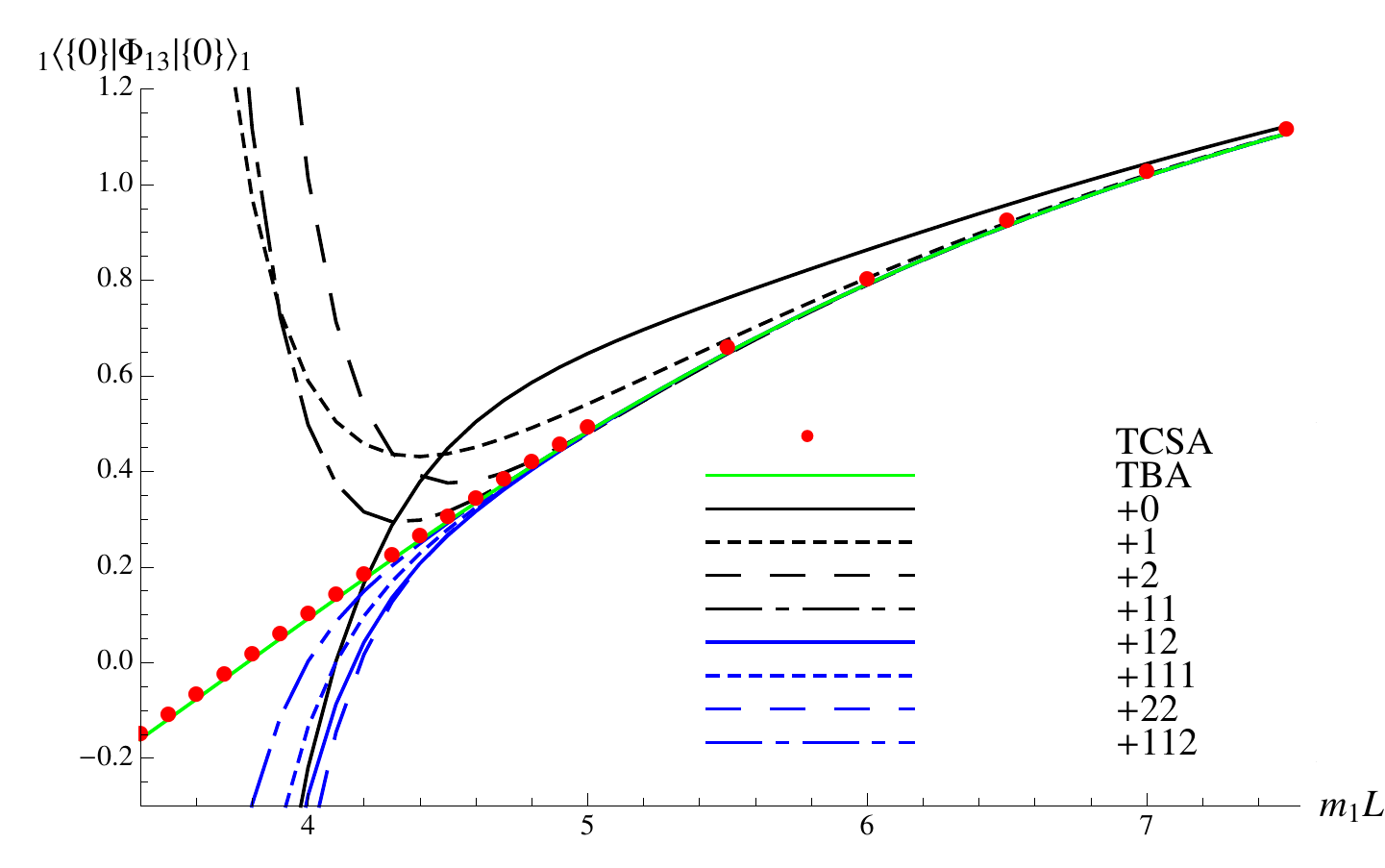}\protect\caption{\label{fig:Phi13_spin0_vev}$\!_{1}\!\!\left\langle \left\{ 0\right\} \right|\Phi_{1,3}\left|\left\{ 0\right\} \right\rangle _{1}$
calculated by the form factor series including different contributions,
RG extrapolated TCSA and TBA.}
\end{figure}
\begin{figure}
\centering{}\includegraphics{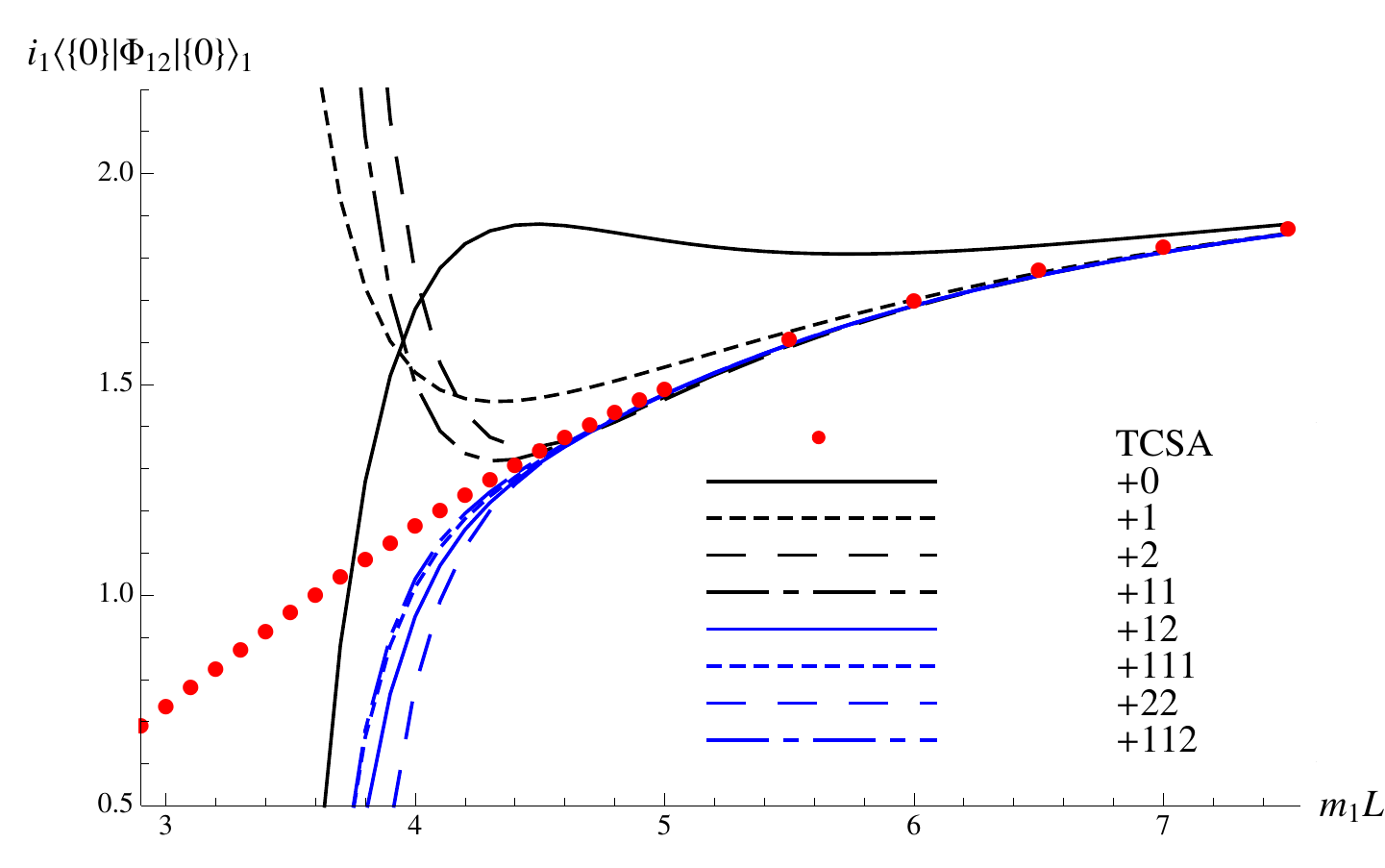}\protect\caption{\label{fig:Phi12_spin0_vev} $i\,_{1}\!\!\left\langle \left\{ 0\right\} \right|\Phi_{1,2}\left|\left\{ 0\right\} \right\rangle _{1}$
calculated by the form factor series including different contributions
and RG extrapolated TCSA.}
\end{figure}
\begin{table}
\begin{centering}
$\begin{array}{|c||c|c|c|c|c|c|c|c|}
\hline m_{1}L & 0 & +1 & +2 & +11 & +12 & +111 & +22 & +112\\
\hline\hline 4 & 0.3114 & -0.4977 & -0.9218 & -0.4061 & 0.3695 & 0.2265 & 0.494 & 9\cdot10^{-2}\\
\hline 4.5 & -0.1529 & -0.1416 & -8\cdot10^{-2} & -2\cdot10^{-2} & 3\cdot10^{-2} & 2\cdot10^{-2} & 3\cdot10^{-2} & 3\cdot10^{-3}\\
\hline 5 & -0.1637 & -6\cdot10^{-2} & 9\cdot10^{-4} & 5\cdot10^{-3} & 3\cdot10^{-3} & 2\cdot10^{-3} & 1\cdot10^{-3} & -3\cdot10^{-5}\\
\hline 5.5 & -0.114 & -3\cdot10^{-2} & 4\cdot10^{-3} & 2\cdot10^{-3} & 3\cdot10^{-4} & 3\cdot10^{-4} & -4\cdot10^{-5} & -3\cdot10^{-5}\\
\hline 6 & -7\cdot10^{-2} & -1\cdot10^{-2} & 2\cdot10^{-3} & 8\cdot10^{-4} & 5\cdot10^{-5} & 5\cdot10^{-5} & -2\cdot10^{-5} & -5\cdot10^{-6}\\
\hline 6.5 & -4\cdot10^{-2} & -6\cdot10^{-3} & 7\cdot10^{-4} & 2\cdot10^{-4} & 7\cdot10^{-6} & 1\cdot10^{-5} & -5\cdot10^{-6} & -9\cdot10^{-7}\\
\hline 7 & -3\cdot10^{-2} & -3\cdot10^{-3} & 2\cdot10^{-4} & 6\cdot10^{-5} & 1\cdot10^{-6} & 2\cdot10^{-6} & -1\cdot10^{-6} & -4\cdot10^{-7}\\
\hline 7.5 & -2\cdot10^{-2} & -1\cdot10^{-3} & 8\cdot10^{-5} & 2\cdot10^{-5} & 1\cdot10^{-7} & 4\cdot10^{-7} & -2\cdot10^{-7} & -1\cdot10^{-7}
\\\hline \end{array}$
\par\end{centering}

\centering{}\protect\caption{\label{tab:Phi13_spin0_diffTCSA}The difference between the evaluatons
of $\!_{1}\!\!\left\langle \left\{ 0\right\} \right|\Phi_{1,3}\left|\left\{ 0\right\} \right\rangle _{1}$
from the RG-extrapolated TCSA and the form factor series, depending
on the multi-particle contributions included in the latter.}
\end{table}
\begin{table}
\begin{centering}
$\begin{array}{|c||c|c|c|c|c|c|c|c|}
\hline m_{1}L & 0 & +1 & +2 & +11 & +12 & +111 & +22 & +112\\
\hline\hline 4 & -0.5274 & -0.3774 & -0.6113 & -0.3511 & 0.2018 & 0.1321 & 0.3664 & 0.1129\\
\hline 4.5 & -0.5513 & -0.139 & -2\cdot10^{-2} & -9\cdot10^{-3} & 2\cdot10^{-2} & 1\cdot10^{-2} & 2\cdot10^{-2} & 4\cdot10^{-3}\\
\hline 5 & -0.3645 & -6\cdot10^{-2} & 1\cdot10^{-2} & 7\cdot10^{-3} & 1\cdot10^{-3} & 1\cdot10^{-3} & 3\cdot10^{-4} & 4\cdot10^{-6}\\
\hline 5.5 & -0.2176 & -3\cdot10^{-2} & 7\cdot10^{-3} & 3\cdot10^{-3} & 2\cdot10^{-4} & 2\cdot10^{-4} & -1\cdot10^{-4} & -3\cdot10^{-5}\\
\hline 6 & -0.1255 & -1\cdot10^{-2} & 3\cdot10^{-3} & 9\cdot10^{-4} & 3\cdot10^{-5} & 5\cdot10^{-5} & -3\cdot10^{-5} & -5\cdot10^{-6}\\
\hline 6.5 & -7\cdot10^{-2} & -7\cdot10^{-3} & 9\cdot10^{-4} & 2\cdot10^{-4} & 5\cdot10^{-6} & 1\cdot10^{-5} & -6\cdot10^{-6} & -2\cdot10^{-6}\\
\hline 7 & -4\cdot10^{-2} & -3\cdot10^{-3} & 3\cdot10^{-4} & 6\cdot10^{-5} & -2\cdot10^{-7} & 1\cdot10^{-6} & -2\cdot10^{-6} & -1\cdot10^{-6}\\
\hline 7.5 & -2\cdot10^{-2} & -1\cdot10^{-3} & 9\cdot10^{-5} & 1\cdot10^{-5} & -1\cdot10^{-6} & -1\cdot10^{-6} & -2\cdot10^{-6} & -2\cdot10^{-6}
\\\hline \end{array}$
\par\end{centering}

\centering{}\protect\caption{\label{tab:Phi12_spin0_diffTCSA}The difference between the evaluations
of $i\,_{1}\!\!\left\langle \left\{ 0\right\} \right|\Phi_{1,2}\left|\left\{ 0\right\} \right\rangle _{1}$
from the RG-extrapolated TCSA and the form factor series, depending
on the multi-particle contributions included in the latter.}
\end{table}

\begin{figure}
\centering{}\includegraphics{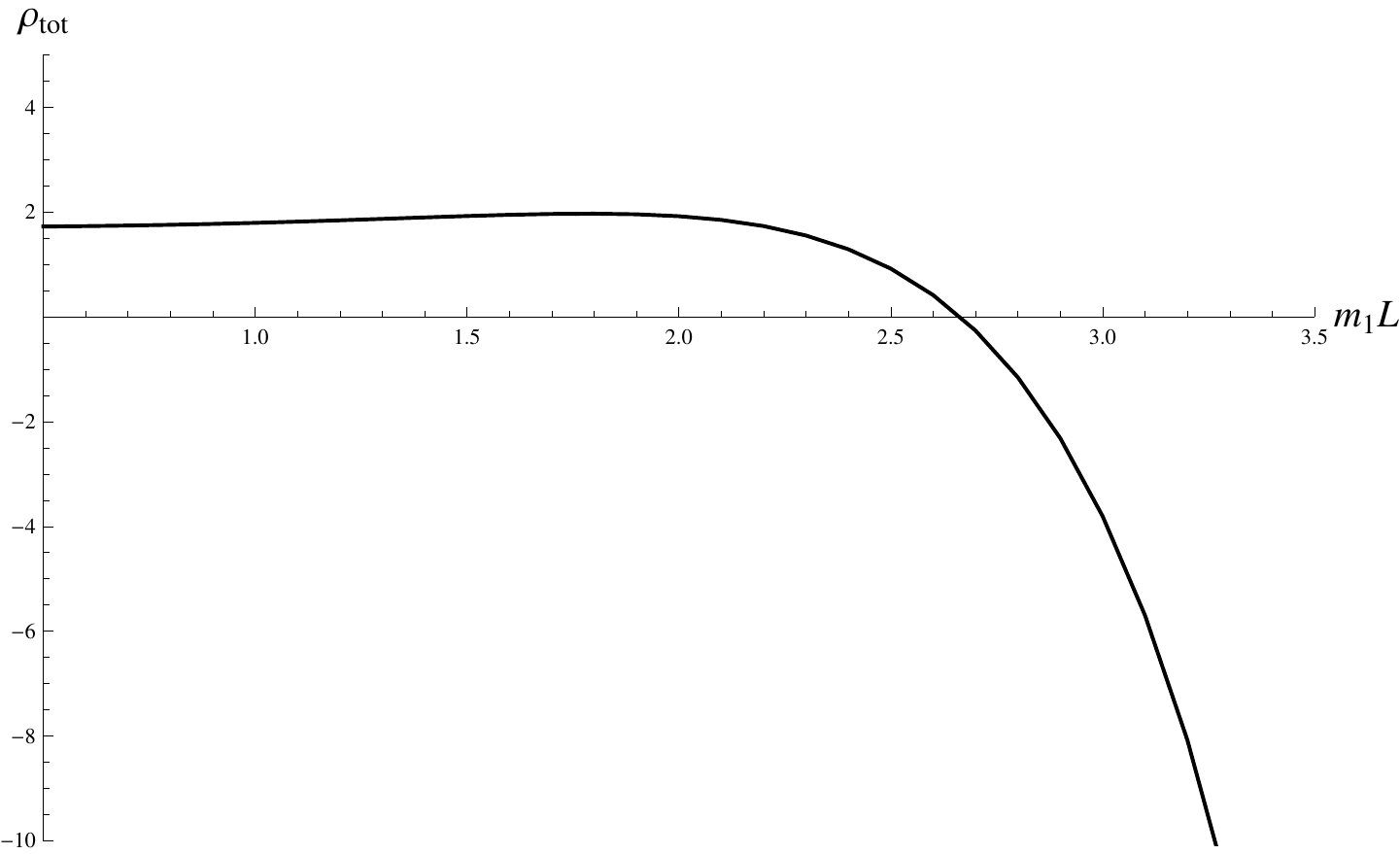}\protect\caption{\label{fig:rho_tot_spin0}$\rho_{tot}$ for s=0}
\end{figure}

\section{Conclusions and outlook \label{sec:Conclusions-and-outlook}}

In this paper we presented a conjecture for the finite volume excited
state expectation values of local operators in integrable quantum
field theories. This conjecture is an extension of an earlier one
\cite{Pozsgay2013} to models with a non-trivial bootstrap structure.
The conjecture was supported by a combination of analytic and numerical
evidence.

An important aspect of our result is that it gives the full specification
of the $\mu$-terms from the excited state TBA. In the previous approach
\cite{Pozsgay2008,Takacs:2011nb}, the determination of these terms
was ambiguous due to the fact that in many cases a particle could
be represented in several ways as bound state of others. The series
(\ref{eq:LMExcited_general}) specifies the $\mu$-term as the one
corresponding to the way the particle is composed in terms of the
singularities entering the excited state TBA equations, which provides
a solution for the case of integrable models where the excited state
TBA system is known. In addition, the series is exact, at least for
values of the volume where it is convergent. 

Unfortunately, for certain excited states this does not cover the
whole volume range, as the TBA singularity configuration undergoes
some rearrangement below some critical value of the volume. In the
series (\ref{eq:LMExcited_general}) this is manifested in a singular
behaviour of the terms, which prevent the extension of the series
below the critical volume. Note that for the trace of the stress-energy
tensor, the desingularized excited TBA gives such an extension. However,
it is presently unclear how to implement the desingularization procedure
directly for the series (\ref{eq:LMExcited_general}), so the description
cannot be extended to other operators. In order to achieve that, one
needs to separate the singularly behaved contributions and re-sum
them to all orders, which is not an obvious task. 

In this context we remark that based on the available studies of excited
state TBA systems, the problematic rearrangement is only expected
to happen for excited states for which some singularities are ``stuck
in the middle'', as opposed to going to the left/right asymptotic
regions logarithmically with decreasing volume. Therefore, for most
states the conjecture is expected to be valid for any value of the
volume; obviously the smaller is the volume, the larger is the number
of terms needed for a given precision. For states with a non-trivial
transition in their singularity structure, the conjecture in its present
form is expected to be valid for values of the volume above the critical
one.

Another interesting issue is to extend the series to theories with
non-diagonal scattering, starting with the series introduced in \cite{Buccheri:2013gta}.
The formalism of finite volume form factors has been partially extended
to these theories \cite{Feher:2011fn,Palmai2013}; unfortunately,
it is exactly the general form of diagonal matrix elements that is
at present not known in full generality. 

Furthermore, the form of the terms in the series (\ref{eq:LMExcited_general})
is very suggestive for an extension to non-diagonal finite volume
matrix elements; however, finding the precise form of such an extension
is still an open question. 

Finally, establishing the relation between the present framework,
and the approaches based on separation of variables \cite{Lukyanov:2000jp,2012JSMTE..10..006G}
or fermionic structures \cite{2011LMaPh..96..325J,2013NuPhB.875..166N,2014IJMPA..2950111N}
would be of interest, as it could lead to more efficient construction
of finite size matrix elements and a deeper understanding of the underlying
principles.

\subsection*{Acknowledgments}

IMSZ is grateful to Roberto Tateo for the help with excited state
TBA numerics and his hospitality in Turin, and also to the INFN and
Campus Hungary Scholarship for financial support of the visit. IMSZ
was also supported by funding from the People Programme (Marie Curie
Actions) of the European Union's Seventh Framework Programme FP7/2007-2013/
under REA Grant Agreement No 317089 (GATIS). BP and GT were supported
by the Momentum grant LP2012-50/2014 of the Hungarian Academy of Sciences.

\appendix

\section{Evaluation of connected diagonal form factors \label{sec:form_factors}}

Here we present a fast and efficient way to evaluate the connected
diagonal form factors (defined in Section \ref{sec:conjecture}),
which proceeds via the so-called symmetric form factors.

\subsection{Form factors of the $T_{2}$ model}

The form factors of primary fields in the $T_{2}$ model were constructed
in \cite{Koubek:1994zp}. For states containing only type-$1$ particles
they are the same as form factors of certain exponential operators
in the sinh-Gordon model with a specific value for the coupling constant.
The form factors of the exponential operator
\begin{equation}
e^{kg\Phi}
\end{equation}
 in the sinh-Gordon model 
\begin{equation}
\mathcal{L}=\frac{1}{2}\partial_{\mu}\Phi\partial^{\mu}\Phi-\lambda\cosh g\Phi\label{eq:sinhG_action}
\end{equation}
have the following form \cite{Koubek:1993ke} 
\begin{equation}
F_{n}^{(k)}\left(\theta_{1},\dots,\theta_{n}\right)=H_{n}\left[k\right]Q_{n}\left(x_{1},\dots,x_{n}\right)\prod_{i<j}\frac{F_{min}\left(\theta_{ij}\right)}{x_{i}+x_{j}}
\end{equation}
where
\begin{equation}
\left[n\right]=\frac{\sin n\pi\frac{B}{2}}{\sin\pi\frac{B}{2}}\qquad B=\frac{2g^{2}}{8\pi+g^{2}}
\end{equation}
and the minimal two-particle form factor is 
\begin{eqnarray}
F_{min}\left(\theta,B\right) & = & \mathcal{N}\exp\left[8\int_{0}^{\infty}\frac{\mathrm{d}x}{x}\frac{\sinh\frac{xB}{4}\sinh\frac{x}{2}\left(1-\frac{B}{2}\right)\sinh\frac{x}{2}}{\sinh^{2}x}\sin^{2}\left[\frac{x(i\pi-\theta)}{2\pi}\right]\right]\nonumber \\
 & = & \mathcal{N}\, e^{I(\theta,B)}\nonumber \\
\mathcal{N} & = & \exp\left[-4\int\frac{\mathrm{d}x}{x}\frac{\sinh\frac{xB}{4}\sinh\frac{x}{2}\left(1-\frac{B}{2}\right)\sinh\frac{x}{2}}{\sinh^{2}x}\right]
\end{eqnarray}
with

\begin{eqnarray}
I\left(\theta,B\right) & = & 8\int_{0}^{\infty}\frac{\mathrm{d}x}{x}\frac{\sinh\frac{xB}{4}\sinh\frac{x}{2}\left(1-\frac{B}{2}\right)\sinh\frac{x}{2}}{\sinh^{2}x}\sin^{2}\left[\frac{x(i\pi-\theta)}{2\pi}\right]=\nonumber \\
 & = & 8\int_{0}^{\infty}\frac{\mathrm{d}x}{x}\frac{\sinh\frac{xB}{4}\sinh\frac{x}{2}\left(1-\frac{B}{2}\right)\sinh\frac{x}{2}}{\sinh^{2}x}\left[N+1-Ne^{-2x}\right]e^{-2Nx}\sin^{2}\left[\frac{x(i\pi-\theta)}{2\pi}\right]+\nonumber \\
 & + & \sum_{k=0}^{N-1}(k+1)\left\{ \log\nu\left[\theta,k+\frac{1}{2}\right]+\log\nu\left[\theta,k+\frac{3}{2}-\frac{B}{4}\right]+\log\nu\left[\theta,k+1+\frac{B}{4}\right]-\right.\nonumber \\
 &  & \left.-\log\nu\left[\theta,k+\frac{3}{2}\right]-\log\nu\left[\theta,k+\frac{1}{2}+\frac{B}{4}\right]-\log\nu\left[\theta,k+1-\frac{B}{4}\right]\right\} 
\end{eqnarray}
and

\begin{equation}
\nu(\theta,a)=1+\frac{(i\pi-\theta)^{2}}{(2\pi)^{2}a^{2}}
\end{equation}
The normalization factors read

\begin{eqnarray}
H_{n} & = & \left(\frac{4\sin\frac{\pi B}{2}}{F_{min}(i\pi,B)}\right)^{n/2}
\end{eqnarray}
Introducing the notations 
\begin{equation}
x_{i}=e^{\theta_{i}}\qquad\mbox{and}\qquad\theta_{ij}=\theta_{i}-\theta_{j}
\end{equation}
the $Q_{n}$ are given in a determinant form

\begin{eqnarray}
Q_{n}(k) & = & \det M_{ij}(k)\nonumber \\
 &  & M_{ij}(k)=\left[i-j+k\right]\sigma_{2i-j}^{(n)}\hfill i,j=1,\dots,n-1
\end{eqnarray}
where the $\sigma_{k}^{(n)}$ are the elementary symmetric polynomials
of order $k$ in the $n$ variables $x_{1},\dots,x_{n}$ defined by
\begin{equation}
\prod_{i=1}^{n}(x+x_{i})=\sum_{k}x^{n-k}\sigma_{k}^{(n)}(x_{1},\dots,x_{n})\label{eq:elsympoldef}
\end{equation}
(this means in particular that $\sigma_{n}^{(k)}=0$ for $k>n$ or
$k<0$). To obtain the form factors of local operators in the $T_{2}$
model it is necessary to set the coupling as 
\begin{equation}
B=-\frac{4}{5}
\end{equation}
Following the procedure in Appendix A of \cite{Pozsgay2006}, the
form factors for type-$2$ particles can be efficiently calculated
with the help of writing the bootstrap fusion in the form

\begin{equation}
F_{2\dots}\left(\theta,\dots\right)=\Gamma_{11}^{2}F_{11\dots}\left(\theta-i\bar{u}_{12}^{1},\theta+i\bar{u}_{12}^{1},\dots\right)\label{eq:fusion}
\end{equation}
where $\bar{u}_{12}^{1}=\frac{\pi}{5}$ and $\Gamma_{11}^{2}=\sqrt{2\tan\left(2\pi/5\right)}$.

The value of $k=1$ corresponds to the operator $\Phi_{1,3}$ and
$k=2$ to the operator $\Phi_{1,2}$. The above form factors are normalized
so that the vacuum expectation value of the field is $1$. To obtain
the conformal normalization used in TCSA, it is necessary to multiply
the form factors with the exact vacuum expectation values known from
\cite{Fateev:1997yg} 
\begin{eqnarray}
\left\langle \Phi_{1,2}\right\rangle  & = & \text{\textminus}2.3251365527\text{·}\text{·}\text{·}\text{\texttimes}im_{1}^{-4/7}\nonumber \\
\left\langle \Phi_{1,3}\right\rangle  & = & 2.2695506880\text{·}\text{·}\text{·}\text{\texttimes}m_{1}^{-6/7}\label{eq:12_13_exactvevs}
\end{eqnarray}

\subsection{Symmetric form factors}

The symmetric form factors are defined as

\begin{eqnarray}
 &  & F_{\underbrace{1,\dots,1}_{n},\underbrace{2,\dots,2}_{m}}^{s}\left(\theta_{1},\dots,\theta_{n},\theta_{n+1},\dots,\theta_{n+m}\right)\nonumber \\
 & = & \lim_{\epsilon\to0}F_{\underbrace{2,\dots,2}_{m},\underbrace{1,\dots,1}_{n},\underbrace{1,\dots,1}_{n},\underbrace{2,\dots,2}_{m}}\left(\theta_{n+m}+i\pi+\epsilon,\dots,\theta_{n+1}+i\pi+\epsilon\right.\nonumber \\
 &  & \left.,\theta_{n}+i\pi+\epsilon,\dots,\theta_{1}+i\pi+\epsilon,\theta_{1},\dots,\theta_{n},\theta_{n+1},\dots,\theta_{n+m}\right)
\end{eqnarray}
where there are $n$ numbers of type-1 particles and $m$ numbers
of type-2 particles. This definition corresponds to a particular specification
for the direction of the limit to the diagonal matrix element. To
compute the above expression, we use fusion (\ref{eq:fusion}) for
type-2 particles, and calculate the limit in terms of a form factor
with $2\left(n+2m\right)$ type-1 particles.

\subsubsection{Denominator and minimal form factors}

The denominator has the following form

\begin{equation}
\prod_{i<j}^{2\left(n+2m\right)}\left(\tilde{x}_{i}-\tilde{x}_{j}\right)
\end{equation}
where 

\begin{equation}
\tilde{x}_{i}=\begin{cases}
-e^{\epsilon}x_{n+2m+1-i} & i\leq n+2m\\
x_{i-n-2m} & i>n+2m
\end{cases}
\end{equation}
To leading order in $\epsilon$, the denominator of the symmetric
form factor takes the form
\begin{equation}
\left(-\epsilon\right)^{n+2m}\left[\prod_{i}^{n+2m}x_{i}\right]\left[\prod_{i<j}^{n+2m}\left(x_{i}^{2}-x_{j}^{2}\right)^{2}\right]\label{eq:symdenom}
\end{equation}
From the minimal form factor part we get an $F_{min}\left(i\pi\right)$
factor for every particle when the rapidities meet with their crossed
version, i.e. a factor of $\left[F_{min}\left(i\pi\right)\right]^{n+2m}$
altogether. 

To simplify the other contribution we use the following relation for
the sinh-Gordon form factors \cite{Fring:1992pt} 
\begin{equation}
F_{min}\left(i\pi+\vartheta\right)F_{min}\left(\vartheta\right)=\frac{\sinh\left(\vartheta\right)}{\sinh\left(\vartheta\right)+i\sin\left(\frac{\pi B}{2}\right)}
\end{equation}
There result for two type-1 particle including the denominator term
is
\begin{equation}
\frac{F_{min}\left(\vartheta_{ij}\right)F_{min}\left(\vartheta_{ji}\right)F_{min}\left(i\pi+\vartheta_{ij}\right)F_{min}\left(i\pi+\vartheta_{ji}\right)}{\left(x_{i}^{2}-x_{j}^{2}\right)^{2}}=\frac{1}{\left(x_{i}^{2}-x_{j}^{2}\right)^{2}+4x_{i}^{2}x_{j}^{2}\sin^{2}\left(\frac{\pi B}{2}\right)}
\end{equation}
The result between one type-1 and a type-2 rapidity is

\begin{equation}
\frac{1}{\left(x_{i}^{2}-x_{j,+}^{2}\right)^{2}+4x_{i}^{2}x_{j,+}^{2}\sin^{2}\left(\frac{\pi B}{2}\right)}\times\frac{1}{\left(x_{i}^{2}-x_{j,-}^{2}\right)^{2}+4x_{i}^{2}x_{j,-}^{2}\sin^{2}\left(\frac{\pi B}{2}\right)}
\end{equation}
where $x_{j,\pm}=x_{j}e^{\pm i\bar{u}_{12}^{1}}$. The result for
rapidities from the same type-2 particle is ($x_{i}=xe^{-i\bar{u}_{12}^{1}}$,$x_{j}=xe^{+i\bar{u}_{12}^{1}}$)
\begin{equation}
\frac{\left[F_{min}\left(\vartheta_{ij}\right)\right]^{2}F_{min}\left(i\pi+\vartheta_{ij}\right)F_{min}\left(i\pi+\vartheta_{ji}\right)}{\left(x_{i}^{2}-x_{j}^{2}\right)^{2}}=\frac{-1}{16x^{4}\sin^{2}\left(\frac{2\pi}{5}\right)}
\end{equation}
The result for rapidities from different type-2 particles is
\begin{eqnarray}
 &  & \frac{\left[F_{min}\left(\theta_{ij}\right)F_{min}\left(i\pi+\theta_{ij}\right)\right]^{2}\left[F_{min}\left(\theta_{ji}\right)F_{min}\left(i\pi+\theta_{ji}\right)\right]^{2}}{\left(x_{i}^{2}-x_{j}^{2}\right)^{4}\left(x_{i}^{2}e^{i2\bar{u}}-x_{j}^{2}e^{-i2\bar{u}}\right)^{2}\left(x_{i}^{2}e^{-i2\bar{u}}-x_{j}^{2}e^{i2\bar{u}}\right)^{2}}\\
 &  & \times\left[F_{min}\left(\theta_{ij}-2i\bar{u}\right)F_{min}\left(i\pi+\theta_{ij}-2i\bar{u}\right)\right]\left[F_{min}\left(\theta_{ij}+2i\bar{u}\right)F_{min}\left(i\pi+\theta_{ij}+2i\bar{u}\right)\right]\nonumber \\
 &  & \times\left[F_{min}\left(\theta_{ji}-2i\bar{u}\right)F_{min}\left(i\pi+\theta_{ji}-2i\bar{u}\right)\right]\left[F_{min}\left(\theta_{ji}+2i\bar{u}\right)F_{min}\left(i\pi+\theta_{ji}+2i\bar{u}\right)\right]\nonumber \\
 & = & \left[\left(x_{i}^{2}-x_{j}^{2}\right)^{2}+4x_{i}^{2}x_{j}^{2}\sin^{2}\left(2\bar{u}\right)\right]^{-2}\left[\left(x_{i}^{2}e^{i2\bar{u}}-x_{j}^{2}e^{-i2\bar{u}}\right)^{2}+4x_{i}^{2}x_{j}^{2}\sin^{2}\left(2\bar{u}\right)\right]^{-1}\nonumber \\
 &  & \times\left[\left(x_{i}^{2}e^{-i2\bar{u}}-x_{j}^{2}e^{i2\bar{u}}\right)^{2}+4x_{i}^{2}x_{j}^{2}\sin^{2}\left(2\bar{u}\right)\right]^{-1}\nonumber 
\end{eqnarray}

\subsubsection{Symmetric polynomial part}

There are $2\left(n+2m\right)$ type-$1$ rapidities due to the fusion,
so the polynomial part of the form factor is a determinant of a $(2\left(n+2m\right)-1)\times(2\left(n+2m\right)-1)$
matrix

\begin{equation}
M_{ij}(k)=\left[i-j+k\right]\sigma_{2i-j}^{(2\left(n+2m\right))}
\end{equation}
In the $\varepsilon\to0$ limit the symmetric polynomials are

\begin{equation}
\begin{array}{c}
\sigma_{l}^{(2p)}\left(x_{1},x{}_{2},\dots x_{p},\right.\\
\left.-e^{\epsilon}x_{1},-e^{\epsilon}x,\dots,-e^{\epsilon}x_{p}\right)
\end{array}\to\begin{cases}
\sigma_{l}^{(2p)}\left(x_{1},x{}_{2},\dots x_{p},-x_{1},-x_{2},\dots,-x_{p}\right)+\mathcal{O}\left(\epsilon\right) & l\quad\mbox{even}\\
\sum_{i=1}^{p}\left(-\epsilon x_{i}\right)\sigma_{l-1}^{(2p-1)}\left(x_{1},x{}_{2},\dots,x_{i},\dots x_{p},\right. & l\quad\mbox{odd}\\
\left.-x_{1},-x_{2},\dots,-x_{i-1},-x_{i+1},\dots,-x_{p}\right)+\mathcal{O}\left(\epsilon^{2}\right)
\end{cases}
\end{equation}
Since every term in the determinant contains $\left(n+2m\right)$
factors of odd symmetrical polynomials, the determinant is proportional
to $\epsilon^{n+2m}$, which exactly cancels the $\epsilon$ powers
in the denominator (\ref{eq:symdenom}). From the definition (\ref{eq:elsympoldef})
of the elementary symmetric polynomials it is easy to show that 
\begin{equation}
\begin{array}{c}
\sigma_{l}^{(2p)}\left(x_{1},x{}_{2},\dots x_{p},\right.\\
\left.-e^{\epsilon}x_{1},-e^{\epsilon}x,\dots,-e^{\epsilon}x_{p}\right)
\end{array}\to\begin{cases}
\left(-1\right)^{l/2}\sigma_{l/2}^{(p)}\left(x_{1}^{2},x_{2}^{2},\dots x_{p}^{2}\right) & l\quad\mbox{even}\\
\sum_{i=1}^{p}\left(-\epsilon x_{i}\right)\left(-1\right)^{\left(l-1\right)/2}\sigma_{\left(l-1\right)/2}^{p-1}\left(x_{1}^{2},x_{2}^{2},\dots,x_{i-1}^{2},\right. & l\quad\mbox{odd}\\
\left.x_{i+1}^{2},\dots x_{p}^{2}\right)
\end{cases}
\end{equation}

\subsubsection{Result for symmetric form factor}

Introducing the following definitions

\begin{eqnarray}
\hat{\sigma}_{l}^{p}\left(x_{1},x{}_{2},\dots x_{p}\right) & = & \begin{cases}
\left(-1\right)^{l/2}\sigma_{l/2}^{(p)}\left(x_{1}^{2},x_{2}^{2},\dots x_{p}^{2}\right) & l\quad\mbox{even}\\
\left(-1\right)^{\left(l-1\right)/2}\sum_{i=1}^{p}x_{i}\sigma_{\left(l-1\right)/2}^{p-1}\left(x_{1}^{2},x_{2}^{2},\dots,x_{i-1}^{2},\right. & l\quad\mbox{odd}\\
\left.x_{i+1}^{2},\dots x_{p}^{2}\right)
\end{cases}\nonumber \\
\hat{Q}_{n+2m}(k) & = & \det\hat{M}_{ij}(k)\qquad\hat{M}_{ij}(k)=\left[i-j+k\right]\hat{\sigma}_{2i-j}^{(n+m)}\nonumber \\
 &  & i,j=1,\dots,2\left(n+2m\right)-1
\end{eqnarray}
and
\begin{eqnarray}
F_{min,denom}\left(x_{1},x_{2}\right) & = & \begin{cases}
F_{min,denom}^{11}\left(x_{1},x_{2}\right) & type-1\leftrightarrow type-1\\
F_{min,denom}^{12}\left(x_{1},x_{2}\right) & type-1\leftrightarrow type-2\\
F_{min,denom}^{22,dif}\left(x_{1},x_{2}\right) & type-2\leftrightarrow type-2
\end{cases}\nonumber \\
F_{min,denom}^{11}\left(x_{1},x_{2}\right) & = & \frac{1}{\left(x_{i}^{2}-x_{j}^{2}\right)^{2}+4x_{i}^{2}x_{j}^{2}\sin^{2}\left(2\bar{u}\right)}\nonumber \\
F_{min,denom}^{12}\left(x_{1},x_{2}\right) & = & F_{min,denom}^{11}\left(x_{1},x_{2}e^{i\bar{u}}\right)\times F_{min,denom}^{11}\left(x_{1},x_{2}e^{-i\bar{u}}\right)\nonumber \\
F_{min,denom}^{22,dif}\left(x_{1},x_{2}\right) & = & \left[F_{min,denom}^{11}\left(x_{1},x_{2}\right)\right]^{2}\times F_{min,denom}^{11}\left(x_{1}e^{i\bar{u}},x_{2}e^{-i\bar{u}}\right)\nonumber \\
 &  & \times F_{min,denom}^{11}\left(x_{1}e^{-i\bar{u}},x_{2}e^{i\bar{u}}\right)\nonumber \\
F_{min,denom}^{22,self}\left(x\right) & = & \frac{-1}{16x^{4}\sin^{2}\left(2\bar{u}\right)}
\end{eqnarray}
the symmetric form factor can be rewritten as:

\begin{eqnarray}
 &  & F_{\underbrace{1,\dots,1}_{n},\underbrace{2,\dots,2}_{m}}^{s}\left(\theta_{1},\dots,\theta_{n},\theta_{n+1},\dots,\theta_{n+m}\right)\\
 & = & \left[k\right]\left(-4\sin\frac{2\pi}{5}\right)^{n+2m}\left(\Gamma_{11}^{2}\right)^{2m}\hat{Q}_{n+2m}\left(x_{1},\dots,x_{n+2m}\right)\nonumber \\
 &  & \times\frac{\prod_{i<j}^{n+m}F_{min,denom}\left(x_{i},x_{j}\right)}{\prod_{i=1}^{n}x_{i}}\prod_{j=n+1}^{m}\frac{F_{min,denom}^{22,self}\left(x_{j}\right)}{x_{j}^{2}}\nonumber 
\end{eqnarray}
For a large number of particles and/or large rapidities this formula
is difficult to evaluate with the required numerical precision, because
the determinant $\hat{Q}$ is badly conditioned (the magnitude of
its matrix elements differ by many orders). For a better precision
it is necessary balance the matrix the following way:

\begin{eqnarray}
\hat{Q}_{n+2m}(k) & = & \left[\hat{\sigma}_{1}^{(n+2m)}/(n+2m)\right]^{2\left(n+2m\right)^{2}-\left(n+2m\right)}\det\widetilde{\hat{M}}_{ij}(k)\nonumber \\
\widetilde{\hat{M}}_{ij}(k) & = & \left[i-j+k\right]\frac{\hat{\sigma}_{2i-j}^{(n+m)}}{\left[\hat{\sigma}_{1}^{(n+2m)}/(n+2m)\right]^{2i-j}}\nonumber \\
 &  & i,j=1,\dots,2\left(n+2m\right)-1
\end{eqnarray}

\subsection{Evaluation of the connected diagonal form factors}

There are two ways to calculate the connected diagonal form factors
using the symmetric form factors. One way is to use the symmetric-connected
relations derived in \cite{Pozsgay2008b} in a recursive manner; this
is a lengthy procedure for form factors with several variables and
not very convenient for numerical calculations. 

However, from the same relations it also follows that the connected
diagonal form factor is the only part of the symmetric form factor
that is fully periodic in all of its variables with period $i\pi$.
This is related to unitarity and crossing invariance, which in theories
with self-conjugate particles take the form 
\begin{eqnarray}
S_{\alpha\beta}(-\theta) & = & S_{\alpha\beta}(\theta)^{-1}\nonumber \\
S_{\alpha\beta}(\theta) & = & S_{\alpha\beta}(i\pi-\theta)
\end{eqnarray}
As a result the kernel functions (\ref{eq:phidef}) have the anti-periodicity
property 
\begin{equation}
\varphi_{\alpha\beta}(\theta+i\pi)=-\varphi_{\alpha\beta}(\theta)
\end{equation}
Applying this property to the connected-symmetric relations of \cite{Pozsgay2008b}
leads to 
\begin{eqnarray}
F_{n}^{c}\left(\theta_{1},\dots,\theta_{n}\right) & = & \frac{1}{2^{n}}\sum_{\alpha_{i}=0,1}F_{n}^{s}\left(\theta_{1}+\alpha_{1}i\pi,\theta_{2}+\alpha_{2}i\pi,\dots,\theta_{n}+\alpha_{n}i\pi\right)\nonumber \\
 & = & \frac{1}{2^{n-1}}\sum_{\alpha_{i}=0,1}F_{n}^{s}\left(\theta_{1},\theta_{2}+\alpha_{2}i\pi,\dots,\theta_{n}+\alpha_{n}i\pi\right)
\end{eqnarray}
which gives a faster and numerically much more stable evaluation.

\bibliographystyle{utphys}
\bibliography{Uber}

\end{document}